\documentclass{sig-alternate}

\conferenceinfo{KDD'11,} {August 21--24, 2011, San Diego, California, USA.} 
\CopyrightYear{2011} 
\crdata{978-1-4503-0813-7/11/08} 
\usepackage{cite}
\usepackage{graphicx}
\usepackage{subfig}
\usepackage{balance}
\usepackage{url}

\newcommand{\argmax}{\operatornamewithlimits{argmax}}

\newdef{defn}{Definition}
\newtheorem{thm}{Theorem}
\newtheorem{prop}{Proposition}

\begin{document}

\title{Benefits of Bias: \\ Towards Better Characterization of Network Sampling}

\numberofauthors{2} 
\author{
\alignauthor Arun S. Maiya\\
       \affaddr{Dept. of Computer Science}\\
       \affaddr{University of Illinois at Chicago}\\
       \email{arun@maiya.net}
\alignauthor Tanya Y. Berger-Wolf\\
       \affaddr{Dept. of Computer Science}\\
       \affaddr{University of Illinois at Chicago}\\
       \email{tanyabw@cs.uic.edu}
}

\newcounter{copyrightbox}

\maketitle
\begin{abstract}
From social networks to P2P systems, network sampling arises in many settings.  We present a detailed study on the nature of biases in network sampling strategies to shed light on how best to sample from networks.  We investigate connections between specific biases and various measures of structural representativeness. We show that certain biases are, in fact, beneficial for many applications, as they ``push'' the sampling process towards inclusion of desired properties.  Finally, we describe how these sampling biases can be exploited in several, real-world applications including disease outbreak detection and market research.
\end{abstract}
 

\section{Introduction and Motivation}
We present a detailed study on the nature of biases in network sampling strategies to shed light on how best to sample from networks.  A \emph{network} is a system of interconnected entities typically represented mathematically as a graph:  a set of vertices and a set of edges among the vertices.  Networks are ubiquitous and arise across numerous and diverse domains.  For instance, many Web-based social media, such as online social networks, produce large amounts of data on interactions and associations among individuals. Mobile phones and location-aware devices produce copious amounts of data on both communication patterns and physical proximity between people. In the domain of biology also, from neurons to proteins to food webs, there is now access to large networks of associations among various entities and a need to analyze and understand these data.

With advances in technology, pervasive use of the Internet, and the proliferation of mobile phones and location-aware devices, networks under study today are not only substantially larger than those in the past, but sometimes exist in a decentralized form (e.g. the network of blogs or the Web itself).  For many networks, their global structure is not fully visible to the public and can only be accessed through ``crawls'' (e.g. online social networks).  These factors can make it prohibitive to analyze or even access these networks in their entirety.  How, then, should one proceed in analyzing and mining these network data?  One approach to addressing these issues is \emph{sampling}:  inference using small subsets of nodes and links from a network.

From epidemiological applications \cite{Cohen2003Efficient} to Web crawling \cite{Boldi2004Do} and P2P search \cite{Stutzbach2009Unbiased}, network sampling arises across many different settings.  In the present work, we focus on a particular line of investigation that is concerned with constructing samples that match critical structural properties of the original network. Such samples have numerous applications in data mining and information retrieval.  In \cite{Krishnamurthy2007Sampling}, for example, structurally-representative samples were shown to be effective in inferring network protocol performance in the larger network and significantly improving the efficiency of protocol simulations.  In Section \ref{sec:applications}, we discuss several additional applications.  Although there have been a number of recent strides in work on network sampling (e.g. \cite{Ahmed2010Timebased,Hubler2008Metropolis,Leskovec2006Sampling,Krishnamurthy2007Sampling}), there is still very much that requires better and deeper understanding.  Moreover, many networks under analysis, although treated as complete, are, in fact, {\em samples} due to limitations in data collection processes.  Thus, a more refined understanding of network sampling is of general importance to network science.    Towards this end, we conduct a detailed study on \emph{network sampling biases}.  There has been a recent spate of work focusing on \emph{problems} that arise from network sampling biases including how and why biases should be avoided \cite{Gjoka2011Walk,Stutzbach2009Unbiased,Costenbader2003Stability,Kurant2010Bias,Achlioptas2005Bias,Stumpf2005Subnets}.  Our work differs from much of this existing literature in that, for the first time in a comprehensive manner, we examine network sampling bias as an \emph{asset to be exploited}.  We argue that  biases of certain sampling strategies can be advantageous if they ``push'' the sampling process towards inclusion of specific properties of interest.\footnote{This is similar to the role of bias in stratified sampling in classical statistics.}  Our main aim in the present work is to identify and understand the connections between specific sampling biases and specific definitions of structural representativeness, so that these biases can be leveraged in practical applications.

\noindent
\textbf{Summary of Findings.}  We conduct a detailed investigation of network sampling biases. We find that bias towards high \emph{expansion} (a concept from expander graphs) offers several unique advantages over other biases such as those toward high degree nodes.  We show both empirically and analytically that such an expansion bias ``pushes'' the sampling process towards new, undiscovered clusters and the discovery of wider portions of the network.  In other analyses, we show that a simple sampling process that selects nodes with many connections from those already sampled is often a reasonably good approximation to directly sampling high degree nodes and locates well-connected (i.e. high degree) nodes significantly faster than most other methods.   We also find that the breadth-first search, a widely-used sampling and search strategy, is surprisingly among the most dismal performers in terms of both discovering the network and accumulating critical, well-connected nodes.  Finally, we describe ways in which some of our findings can be exploited in several important applications including disease outbreak detection and market research.  A number of these aforementioned findings are surprising in that they are in stark contrast to conventional wisdom followed in much of the existing literature (e.g. \cite{Najork2001Breadthfirst,Mislove2007Measurement,Adamic2001Search,Cohen2003Efficient,Kurant2010Bias}).

\section{Related Work}
\label{sec:relatedwork}
Not surprisingly, network sampling arises across many diverse areas.  Here, we briefly describe some of these different lines of research.

\noindent
\textbf{Network Sampling in Classical Statistics.}  The concept of sampling networks first arose to address scenarios where one needed to study hidden or difficult-to-access populations (e.g. illegal drug users, prostitutes).  For recent surveys, one might refer to \cite{Kolaczyk2009Statistical,Frank2005Models}. The work in this area focuses almost exclusively on acquiring unbiased estimates related to variables of interest attached to each network node.  The present work, however, focuses on inferring properties related to the \emph{network itself} (many of which are not amenable to being fully captured by simple attribute frequencies).  Our work, then, is much more closely related to \emph{representative subgraph sampling}.

\noindent
\textbf{Representative Subgraph Sampling.}  In recent years, a number of works have focused on \emph{representative subgraph sampling}:  constructing samples in such a way that they are condensed representations of the original network (e.g. \cite{Lee2006Statistical,Ahmed2010Timebased,Hubler2008Metropolis,Leskovec2006Sampling,Krishnamurthy2007Sampling}).  Much of this work focuses on how best to produce a ``universal'' sample representative of \emph{all} structural properties in the original network.  By contrast, we subscribe to the view that no single sampling strategy may be appropriate for all applications.  Thus, our aim, then, is to better understand the \emph{biases} in specific sampling strategies to shed light on how best to leverage them in practical applications.

\noindent
\textbf{Unbiased Sampling.}  There has been a relatively recent spate of work (e.g. \cite{Stutzbach2009Unbiased,Gjoka2011Walk,Henzinger2000Nearuniform}) that focuses on constructing uniform random samples in scenarios where nodes cannot be easily drawn randomly (e.g. settings such as the Web where nodes can only be accessed through crawls).  These strategies, often based on modified random walks, have been shown to be effective for various frequency estimation problems (e.g. inferring the proportion of pages of a certain language in a Web graph \cite{Henzinger2000Nearuniform}).  However, as mentioned above, the present work focuses on using samples to infer structural (and functional) properties of the \emph{network itself}.  In this regard, we found these unbiased methods to be less effective during preliminary testing.  Thus, we do not consider them and instead focus our attention on other more appropriate sampling strategies (such as those mentioned in  \emph{representative subgraph sampling}).

\noindent
\textbf{Studies on Sampling Bias.} Several studies have investigated \emph{biases} that arise from various sampling strategies (e.g. \cite{Lakhina2003Sampling,Costenbader2003Stability,Kurant2010Bias,Achlioptas2005Bias,Stumpf2005Subnets}).  For instance, \cite{Stumpf2005Subnets} showed that, under the simple sampling strategy of picking nodes at random from a scale-free network (i.e. a network whose degree distribution follows the power law), the resultant subgraph sample will \emph{not} be scale-free.  The authors of \cite{Lakhina2003Sampling,Achlioptas2005Bias} showed the converse is true under traceroute sampling.  Virtually all existing results on network sampling bias focus on its negative aspects.  By contrast, we focus on the \emph{advantages} of certain biases and ways in which they can be exploited in network analysis.

\noindent
\textbf{Property Testing.}  Work on sampling exists in the fields of combinatorics and graph theory and is centered on the notion of \emph{property testing} in graphs \cite{Lovasz2009Very}.  Properties such as those typically studied in graph theory, however, may be less useful for the analysis of \emph{real-world} networks (e.g. the exact meaning of, say, $k$-colorability \cite{Lovasz2009Very} within the context of a social network is unclear).  Nevertheless, theoretical work on property testing in graphs is excellently surveyed in \cite{Lovasz2009Very}.

\noindent
\textbf{Other Areas.}  Decentralized search (e.g. searching unstructured P2P networks) and Web crawling can both be framed as network sampling problems, as both involve making decisions from subsets of nodes and links from a larger network.  Indeed, network sampling itself can be viewed as a problem of information retrieval, as the aim is to seek out a subset of nodes that either individually or collectively match some criteria of interest.  Several of the sampling strategies we study in the present work, in fact, are graph search algorithms (e.g. breadth-first search).   Thus, a number of our findings discussed later have implications for these research areas (e.g. see \cite{Maiya2010Expansion}).  For reviews on decentralized search both in the contexts of complex networks and P2P systems, one may refer to \cite{Kleinberg2006Complex} and \cite{Tsoumakos2006Analysis}, respectively.  For examples of connections between Web crawling and network sampling, see \cite{Najork2001Breadthfirst,Cho1998Efficient,Boldi2004Do}.

\section{Preliminaries}

\subsection{Notations and Definitions}
We now briefly describe some notations and definitions used throughout this paper.

\begin{defn}
\label{defn:network}
$G=(V,\,E)$ is a \emph{network} or \emph{graph} where $V$ is set of vertices and $E \subseteq V \times V$ is a set of edges.  
\end{defn}

\begin{defn}
\label{defn:sample}
A \emph{sample} $S$ is a subset of vertices, $S \subset V$.
\end{defn}

\begin{defn}
\label{defn:neighborhood}
$N(S)$ is the \emph{neighborhood} of $S$ if $N(S)=\{w \in V-S: \; \exists v \in S \; s.t. \; (v,\,w) \in E\}$. 
\end{defn}

\begin{defn}
\label{defn:inducedsubgraph}
$G_S$ is the \emph{induced subgraph} of $G$ based on the sample $S$ if $G_S = (S,\,E_S)$ where the vertex set is $S \subset V$ and the edge set is $E_S = (S \times S) \cap E$.  The induced subgraph of a sample may also be referred to as a \emph{subgraph sample}.
\end{defn}

\subsection{Datasets}

We study sampling biases in a total of twelve different networks:  a power grid (PowerGrid \cite{Watts1998Collective}), a Wikipedia voting network (WikiVote \cite{LeskovecStanford}), a PGP trust network (PGP \cite{Boguna2004Models}), a citation network (HEPTh \cite{LeskovecStanford}), an email network (Enron \cite{LeskovecStanford}), two co-authorship networks (CondMat \cite{LeskovecStanford} and AstroPh \cite{LeskovecStanford}), two P2P file-sharing networks (Gnutella04 \cite{LeskovecStanford} and Gnutella31 \cite{LeskovecStanford}),  two online social networks (Epinions \cite{LeskovecStanford} and Slashdot \cite{LeskovecStanford} ), and a product co-purchasing network (Amazon \cite{LeskovecStanford}).  These datasets were chosen to represent a rich set of diverse networks from different domains.  This diversity allows a more comprehensive study of network sampling and thorough assessment of the performance of various sampling strategies in the face of varying network topologies.  Table \ref{tab:datasets} shows characteristics of each dataset.  All networks are treated as undirected and unweighted.

\begin{table} [th]
\centering
\footnotesize
\begin{tabular}{l|c|c|c|c|c} \hline \hline
Network        &  N     & D                        & PL   & CC       & AD \\ \hline
 PowerGrid        & 4941   & 0.0005                   & 19   & 0.11     & 2.7 \\
 WikiVote       & 7066   & 0.004                   & 3.3     & 0.21     & 28.5 \\
 PGP          & 10,680 & 0.0004                   & 7.5    & 0.44     & 4.6 \\
 Gnutella04          & 10,876 & 0.0006                   & 4.6    & 0.01     & 7.4 \\
AstroPh           & 17,903 & 0.0012                   & 4.2  & 0.67     & 22.0 \\
CondMat           & 21,363 & 0.0004                   & 5.4  & 0.70     & 8.5 \\
 HEPTh             & 27,400 & 0.0009                   &  4.3      & 0.34     & 25.7 \\
 Enron             & 33,696 & 0.0003                   & 4.0    & 0.71     & 10.7 \\
 Gnutella31          & 62,561  & 0.00008                  & 5.9    & 0.01      & 4.7 \\
 Epinions          & 75,877  & 0.0001                  & 4.3    & 0.26      & 10.7 \\
 Slashdot         & 82,168  & 0.0001                  & 4.1    & 0.10     & 12.2 \\ 
 Amazon         & 262,111  & 0.00003                  & 8.8    & 0.43     & 6.9 \\ \hline \hline
\end{tabular}
\caption{{\footnotesize Network Properties.  \textbf{Key:}  \emph {N= \# of nodes, D= density, PL = characteristic path length, CC = local clustering coefficient, AD = average degree.} }}
\label{tab:datasets}
\vskip -0.15in
\end{table}

\section{Network Sampling}
\label{sec:sampling}

In the present work, we focus on a particular class of sampling strategies, which we refer to as \emph{link-trace sampling}.  In \emph{link-trace sampling}, the next node selected for inclusion into the sample is always chosen from among the set of nodes directly connected to those already sampled.  In this way, sampling proceeds by tracing or following links in the network.  This concept can be defined formally.

\begin{defn}
\label{defn:linktracesampling}
Given an integer $k$ and an initial node (or seed) $v \in V$ to which $S$ is initialized (i.e. $S=\{v\}$), a \emph{link-trace sampling} algorithm, $\mathcal{A}$, is a process by which nodes are iteratively selected from among the current neighborhood $N(S)$ and added to $S$ until $|S|=k$.    
\end{defn}

\emph{Link-trace sampling} may also be referred to as \emph{crawling} (since links are ``crawled'' to access nodes) or viewed as \emph{online} sampling (since the network $G$ reveals itself iteratively during the course of the sampling process).  The key advantage of sampling through link-tracing, then, is that complete access to the network in its entirety is \emph{not} required.  This is beneficial for scenarios where the network is either large (e.g. an online social network), decentralized (e.g. an unstructured P2P network), or both (e.g. the Web).  

As an aside, notice from Definition \ref{defn:linktracesampling} that we have implicitly assumed that the neighbors of a given node can be obtained by visiting that node during the sampling process (i.e. $N(S)$ is known).  This, of course, accurately characterizes most real scenarios.  For instance, neighbors of a Web page can be gleaned from the  hyperlinks on a visited page and neighbors of an individual in an online social network can be acquired by viewing (or ``scraping'') the friends list. 

Having provided a general definition of \emph{link-trace sampling}, we must now address \emph{which} nodes in $N(S)$ should be preferentially selected at each iteration of the sampling process.  This choice will obviously directly affect the properties of the sample being constructed.  We study seven different approaches - all of which are quite simple yet, at the same time, ill-understood in the context of real-world networks.

\noindent
\textbf{Breadth-First Search (BFS).} Starting with a single seed node, the BFS explores the neighbors of visited nodes.  At each iteration, it traverses an unvisited neighbor of the \emph{earliest} visited node \cite{Cormen2003Introduction}.  In both \cite{Kurant2010Bias} and \cite{Najork2001Breadthfirst}, it was empirically shown that BFS is biased towards high-degree and high-PageRank nodes.  BFS is used prevalently to crawl and collect networks (e.g. \cite{Mislove2007Measurement}).  

\noindent
\textbf{Depth-First Search (DFS).}  DFS is similar to BFS, except that, at each iteration, it visits an unvisited neighbor of the most \emph{recently} visited node \cite{Cormen2003Introduction}.

\noindent
\textbf{Random Walk (RW).} A random walk simply selects the next hop uniformly at random from among the neighbors of the current node \cite{Lovasz1994Random}. 

\noindent
\textbf{Forest Fire Sampling (FFS).}  FFS, proposed in \cite{Leskovec2006Sampling}, is essentially a probabilistic version of BFS.  At each iteration of a BFS-like process, a neighbor $v$ is only explored according to some ``burning'' probability $p$.  At $p=1$, FFS is identical to BFS.  We use $p=0.7$, as recommended in \cite{Leskovec2006Sampling}.

\noindent
\textbf{Degree Sampling (DS).}  The DS strategy involves greedily selecting the node $v \in N(S)$ with the highest degree (i.e. number of neighbors).  A variation of DS was analytically and empirically studied as a P2P search algorithm in \cite{Adamic2001Search}.  Notice that, in order to select the node $v\in N(S)$ with the highest degree, the process must know $|N(\{v\})|$ for each $v \in N(S)$.  That is, knowledge of $N(N(S))$ is required at each iteration.  As noted in \cite{Adamic2001Search}, this requirement is acceptable for some domains such as P2P networks and certain social networks.  The DS method is also feasible in scenarios where 1) one is interested in efficiently ``downsampling'' a network to a connected subgraph, 2) a crawl is repeated and history of the last crawl is available, or 3) the proportion of the network accessed to construct a sample is less important.

\noindent
\textbf{SEC (Sample Edge Count).}  Given the currently constructed sample $S$, how can we select a node $v \in N(S)$ with the highest degree \emph{without} having knowledge of $N(N(S))$?  The SEC strategy tracks the links from the currently constructed sample $S$ to each node $v \in N(S)$ and selects the node $v$ with the most links from $S$.  In other words, we use the degree of $v$ in the induced subgraph of $S \cup \{v\}$ as an approximation of the degree of $v$ in the original network $G$.   Similar approaches have been employed as part of Web crawling strategies with some success (e.g.  \cite{Cho1998Efficient}). 

\noindent
\textbf{XS (Expansion Sampling).} The XS strategy is based on the concept of expansion from work on expander graphs and seeks to greedily construct the sample with the maximal expansion:  $\argmax_{S:\,|S|=k} \frac{|N(S)|}{|S|}$, where $k$ is the desired sample size \cite{Maiya2010Sampling,Hoory2006Expander}.  At each iteration, the next node $v$ selected  for inclusion in the sample is chosen based on the expression:  $$\argmax_{v \in N(S)} |N(\{v\})-(N(S) \cup S)|.$$   Like the DS strategy, this approach utilizes knowledge of $N(N(S))$. In Sections \ref{sec:rep.reach} and \ref{sec:biases.xs}, we will investigate in detail the effect of this expansion bias on various properties of constructed samples.

\section{Evaluating Representativeness}
\label{sec:rep}

What makes one sampling strategy ``better'' than another?  In computer science, ``better'' is typically taken to be structural \emph{representativeness} (e.g. see \cite{Krishnamurthy2007Sampling,Hubler2008Metropolis,Leskovec2005Graphs}).  That is, samples are considered better if they are more representative of structural properties in the original network.  There are, of course, numerous structural properties from which to choose, and, as correctly observed by Ahmed et al. \cite{Ahmed2010Reconsidering},  it is not always clear which should be chosen.  Rather than choosing arbitrary structural properties as measures of representativeness, we select specific measures of representativeness that we view as being potentially useful for real applications.  We divide these measures (described below) into three categories:  Degree, Clustering, and Reach.    For each sampling strategy, we generate 100 samples using randomly selected seeds, compute our measures of representativeness on each sample, and plot the average value as sample size grows.  (Standard deviations of computed measures are discussed in Section \ref{sec:rep.seedsensitivity}.  Applications for these measures of representativeness are discussed later in Section \ref{sec:applications}.)  Due to space limitations and the large number of networks evaluated, for each evaluation measure, we only show results for two datasets that are illustrative of general trends observed in all datasets.  However, full results are available as supplementary material.\footnote{Supplementary material for this paper is available at:\\  \url{http://arun.maiya.net/papers/supp-netbias.pdf}}

\subsection{Degree}
\label{sec:rep.degree}

The degrees (numbers of neighbors) of nodes in a network is a fundamental and well-studied property.    In fact, other graph-theoretic properties such as the average path length between nodes can, in some cases, be viewed as byproducts of degree (e.g. short paths arising from a small number of highly-connected hubs that act as conduits \cite{Barabasi1999Emergence}).  We study two different aspects of degree (with an eye towards real-world applications, discussed in Section \ref{sec:applications}).

\subsubsection{Measures}
\label{sec:rep.degree.measures}
\noindent
\textbf{Degree Distribution Similarity ({\sc DistSim}).}  We take the degree sequence of the sample and compare it to that of the original network using the two-sample Kolmogorov-Smirnov (K-S) D-statistic \cite{Leskovec2006Sampling}, a distance measure.  Our objective here is to measure the agreement between the two degree distributions in terms of both shape and location.  Specifically, the D-statistic is defined as $D = \max_x\{|F(x) - F_S(x)|\}$, where $x$ is the range of node degrees, and $F$ and $F_S$ are the cumulative degree distributions for $G$ and $G_S$, respectively \cite{Leskovec2006Sampling}.  We compute the distribution similarity by subtracting the K-S distance from one.

\noindent
\textbf{Hub Inclusion ({\sc Hubs}).} In several applications, one cares less about matching the \emph{overall} degree distribution and more about accumulating the highest degree nodes into the sample quickly (e.g. immunization strategies \cite{Cohen2003Efficient}).  For these scenarios, sampling is used as a tool for information retrieval.  Here, we evaluate the extent to which sampling strategies accumulate hubs (i.e. high degree nodes) quickly into the sample.  As sample size grows, we track the proportion of the top $K$ nodes accumulated by the sample. For our tests, we use $K=100$.

\subsubsection{Results}
\label{sec:rep.degree.results}
Figure \ref{fig:rep.degree} shows the \emph{degree distribution similarity} ({\sc DistSim}) and \emph{hub inclusion} ({\sc Hubs}) for the Slashdot and Enron datasets.  Note that the SEC and DS strategies, both of which are biased to high degree nodes, perform best on \emph{hub inclusion} (as expected), but are the \emph{worst} performers on the {\sc DistSim} measure (which is also a direct result of this bias). (The XS strategy exhibits a similar trend but to a slightly lesser extent.)   On the other hand, strategies such as BFS, FFS, and RW tend to perform better on {\sc DistSim}, but worse on {\sc Hubs}.  For instance, the DS and SEC strategies locate the majority of the top 100 hubs with sample sizes less than $1\%$ in some cases.  BFS and FFS require sample sizes of over $10\%$ (and the performance differential is larger when locating hubs ranked higher than $100$).  More importantly, no strategy performs best on \emph{both} measures.  This, then,  suggests a tension between goals:  constructing small samples of the most well-connected nodes is in conflict with producing small samples exhibiting representative degree distributions.  More generally, when selecting sample elements, choices resulting in gains for one area can result in losses for another. Thus, these choices must be made in light of how samples will be used - a subject we discuss in greater depth in Section \ref{sec:applications}.  We conclude this section by briefly noting that the trend observed for SEC seems to be somewhat dependent upon the quality and number of hubs actually present in a network (relative to the size of the network, of course). That is, SEC matches DS more closely as degree distributions exhibit longer and denser tails (as shown in Figure \ref{fig:rep.dd}).  We will revisit this in Section \ref{sec:biases.sec}.  (Other strategies are sometimes affected similarly, but the trend is much less consistent.)  In general, we find SEC best matches DS performance on many of the social networks (as opposed to technological networks such as the PowerGrid with few ``good'' hubs, lower average degree, and longer path lengths).  However, further investigation is required to draw firm conclusions on this last point.

\begin{figure}[htb]
  \centering
  \subfloat[Slashdot {\sc (DistSim)} ] {\label{fig:rep.distsim.slashdot}\includegraphics[width=0.2\textwidth]{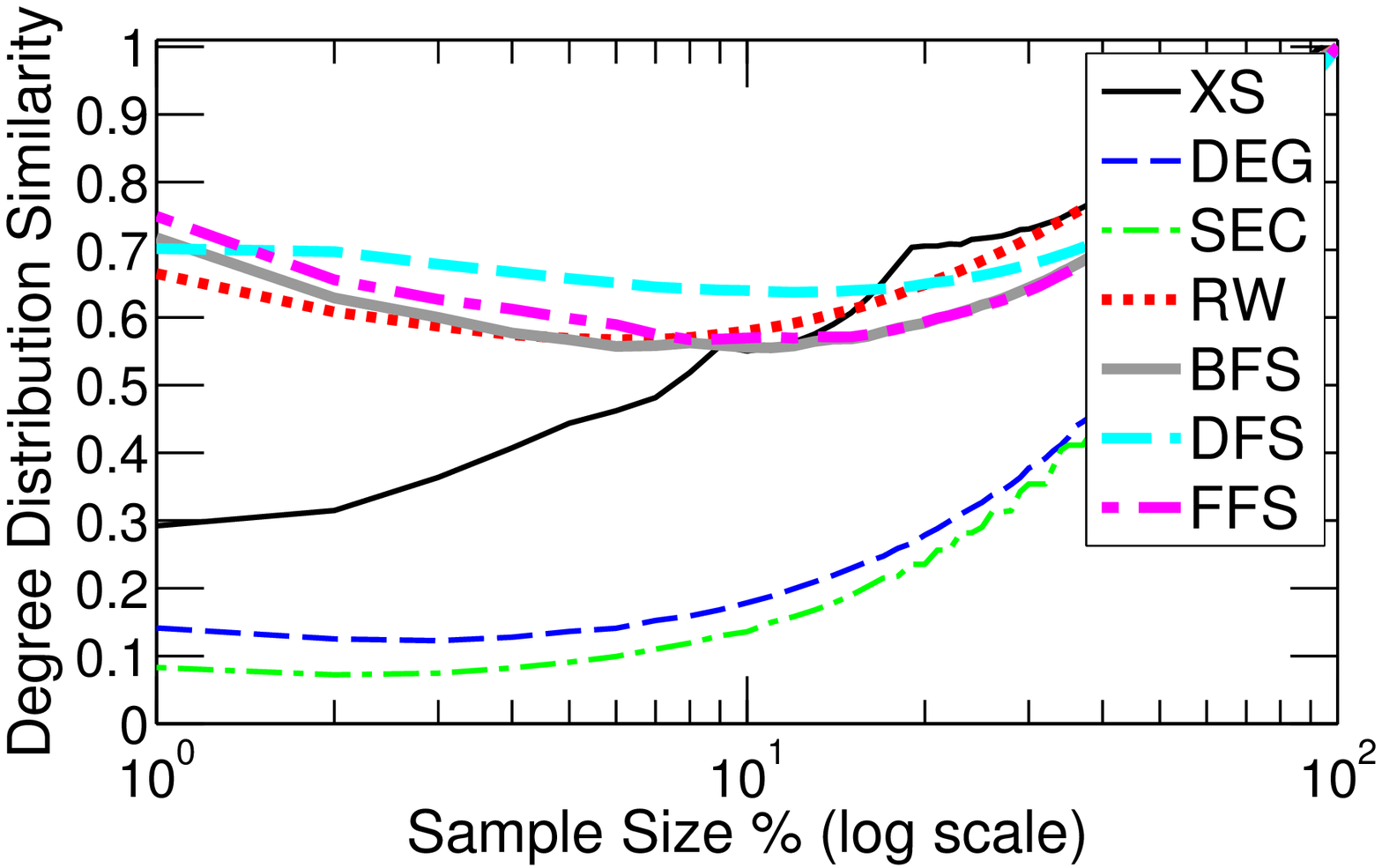}} \vspace{.05cm}
  \subfloat[Enron {\sc (DistSim)}]{\label{fig:rep.distsim.enron}\includegraphics[width=0.2\textwidth]{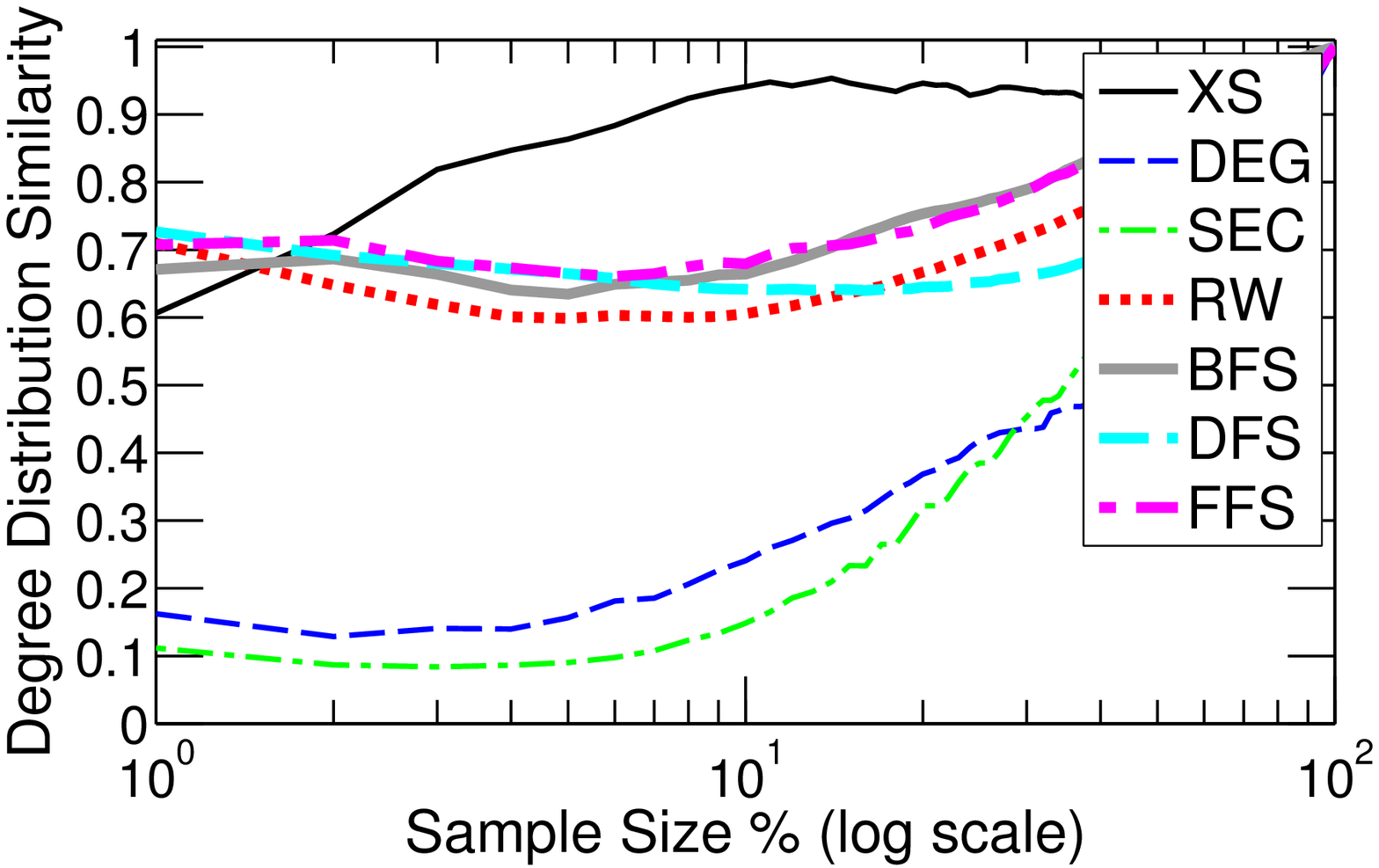}} \\ \vskip -0.01in
  \subfloat[Slashdot {\sc (Hubs)}] {\label{fig:rep.hubs.slashdot}\includegraphics[width=0.2\textwidth]{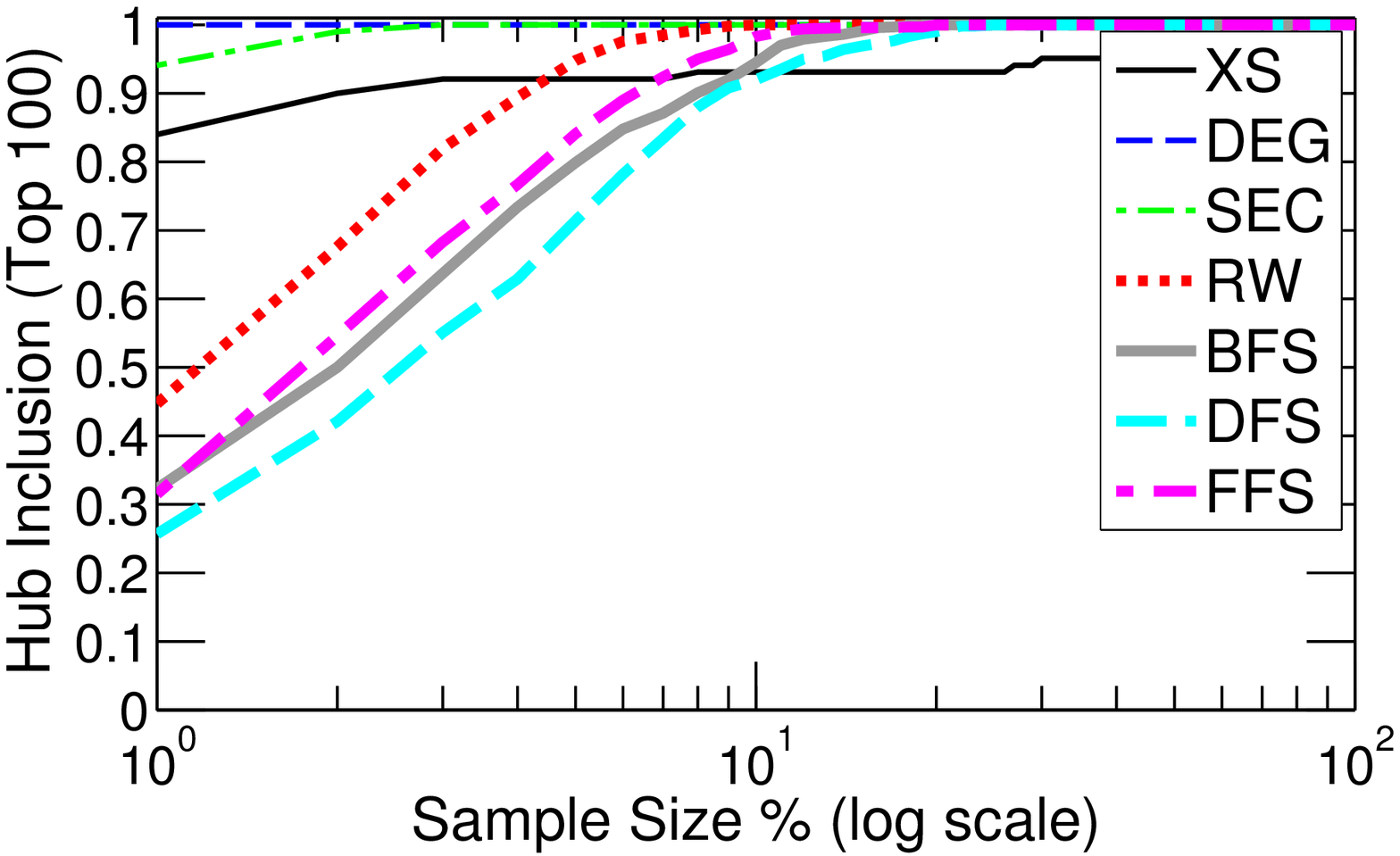}} \vspace{.05cm}
  \subfloat[Enron {\sc (Hubs)} ]{\label{fig:rep.hubs.enron}\includegraphics[width=0.2\textwidth]{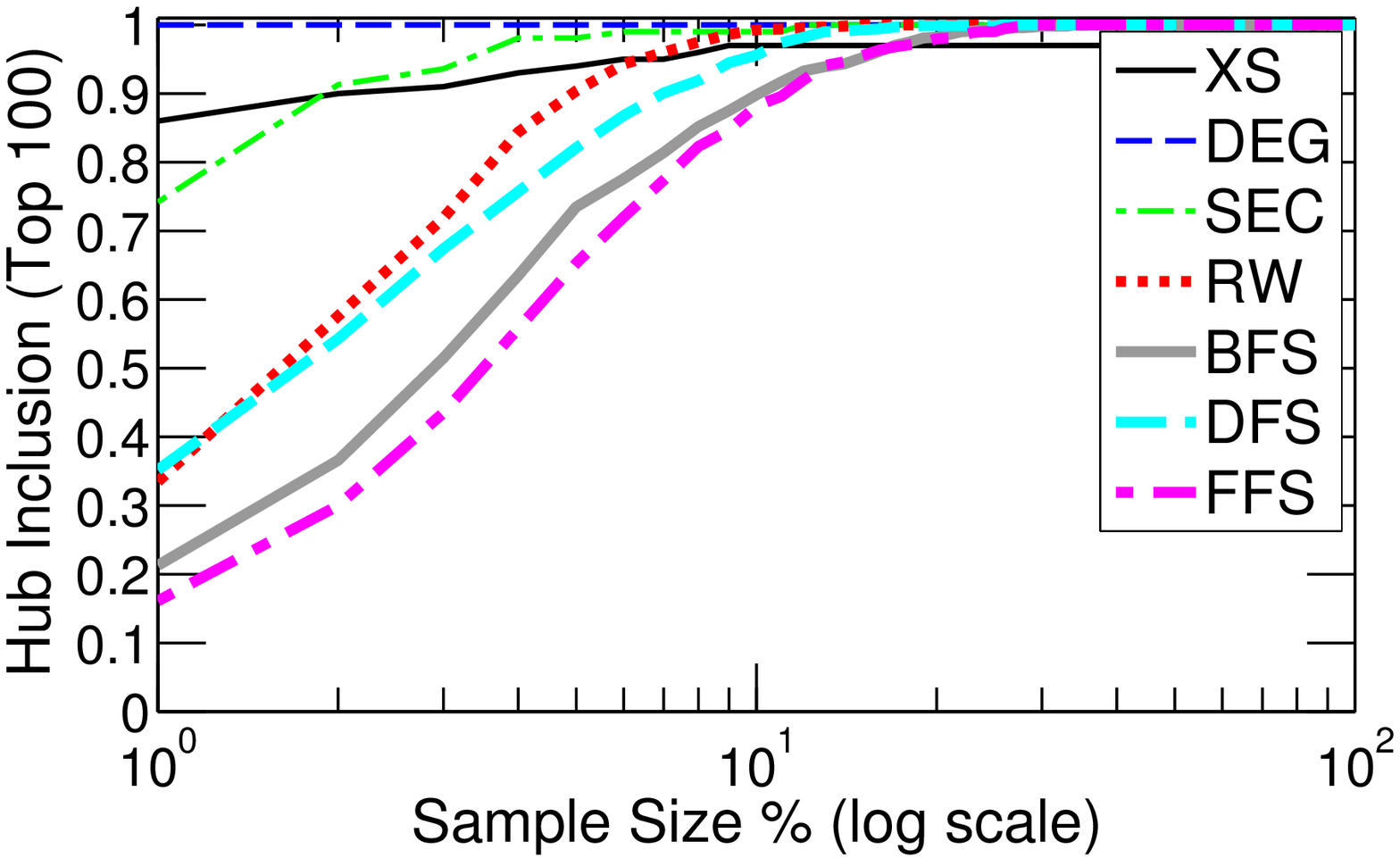}}
\caption{Evaluating {\sc (DistSim)} and  {\sc (Hubs)}.  Results for remaining networks are similar.  }
  \label{fig:rep.degree}
\vskip -0.15in
\end{figure}

\begin{figure}[htb]
  \centering
  \subfloat[Epinions {\sc (Hubs)} ] {\label{fig:rep.hubs.epinions}\includegraphics[width=0.2\textwidth]{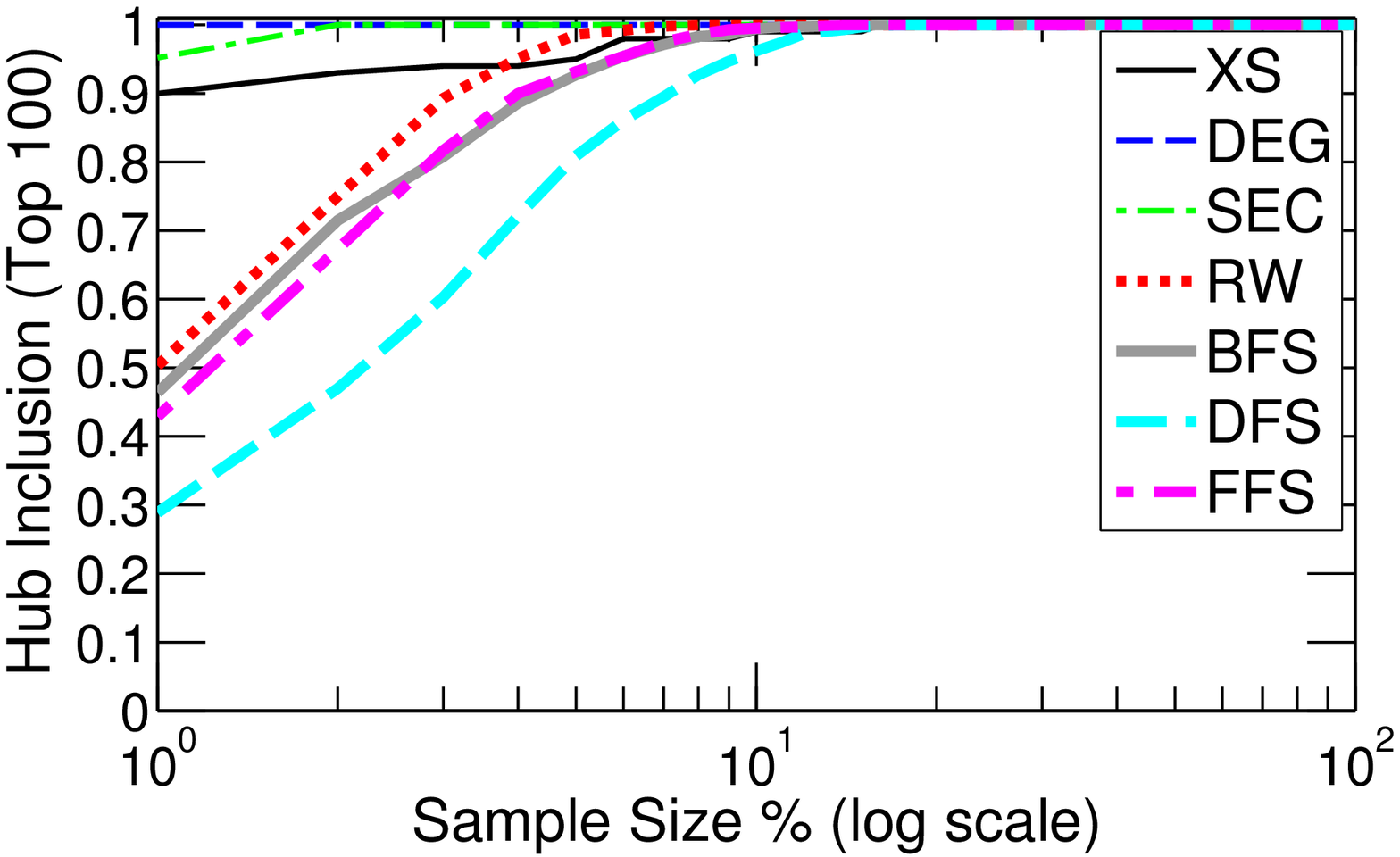}} \vspace{.05cm}
  \subfloat[Epinions {\sc (DD)}]{\label{fig:rep.dd.epinions}\includegraphics[width=0.2\textwidth]{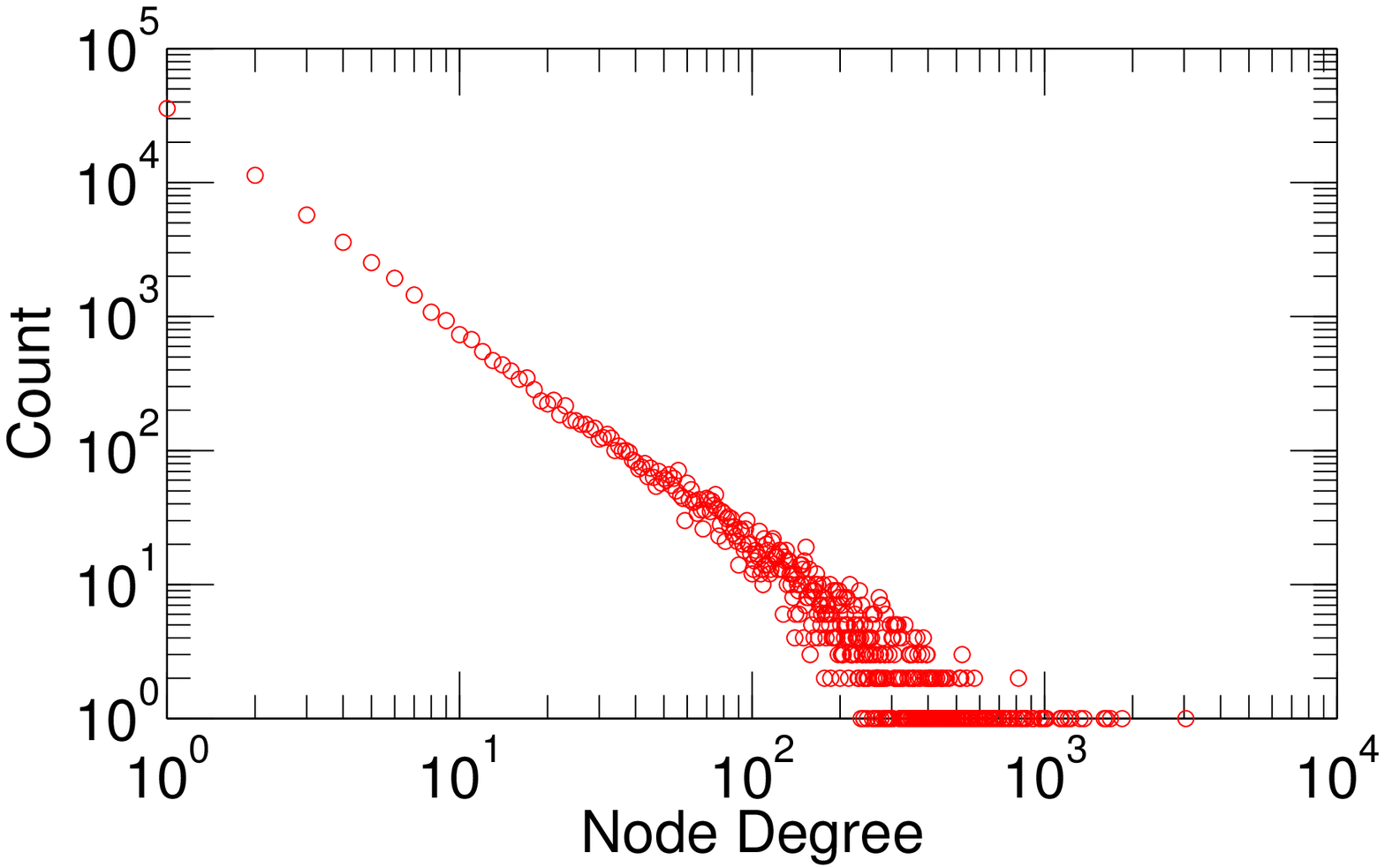}} \\ \vskip -0.01in
  \subfloat[Gnutella31 {\sc (Hubs)}] {\label{fig:rep.hubs.gnutella31}\includegraphics[width=0.2\textwidth]{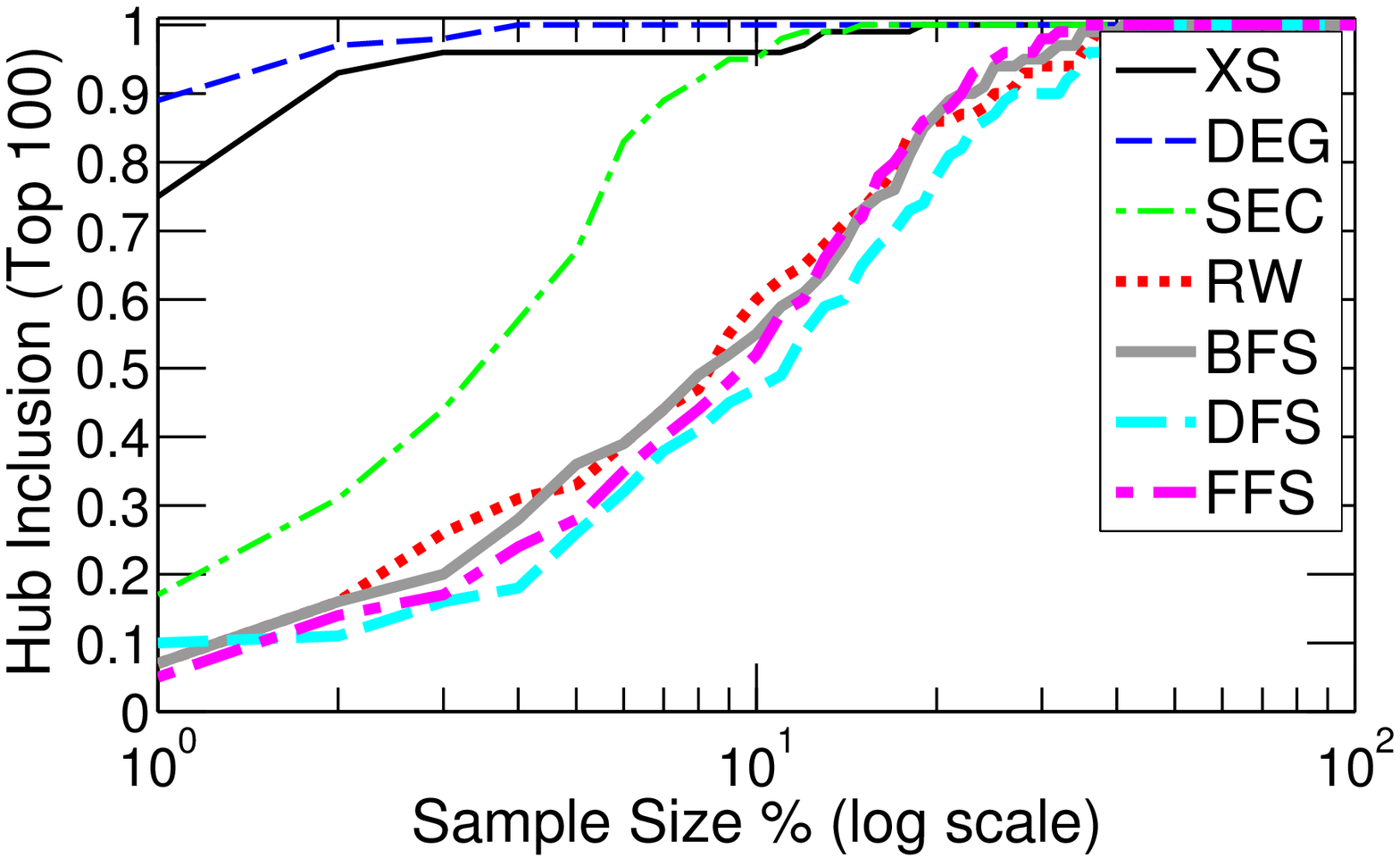}} \vspace{.05cm}
  \subfloat[Gnutella31 {\sc (DD)} ]{\label{fig:rep.dd.gnutella31}\includegraphics[width=0.2\textwidth]{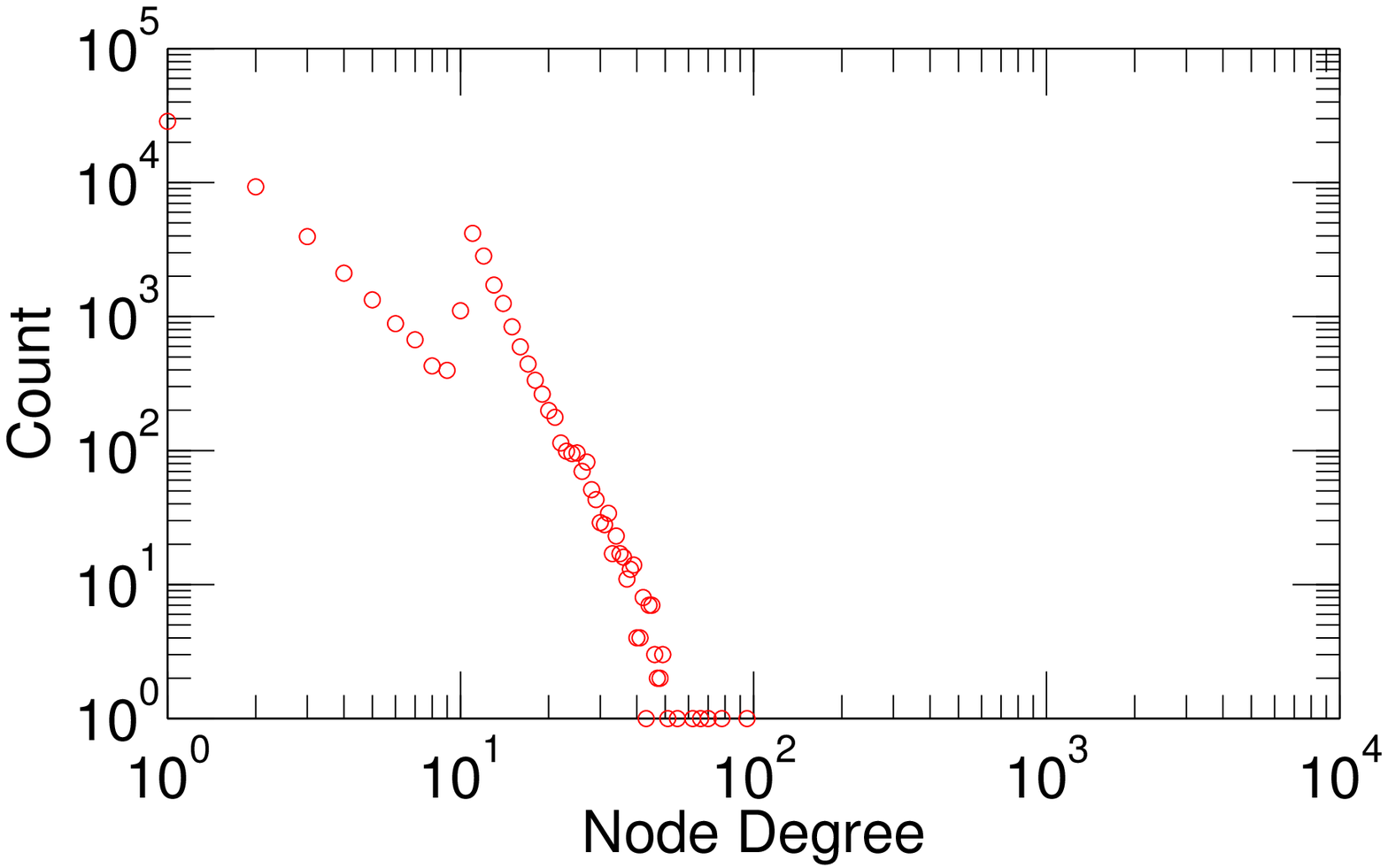}}
\caption{ {\footnotesize Performance of SEC on {\sc Hubs} (shown in green on left) is observed to be dependent on the tail of the degree distributions ({\sc DD}).  SEC matches DS more closely when more and better quality hubs are present.  For {\sc Hubs}, SEC generally performs best on the social networks evaluated.} }
  \label{fig:rep.dd}
\vskip -0.15in
\end{figure}

\subsection{Clustering}
\label{sec:rep.clustering}
Many real-world networks, such as social networks, exhibit a much higher level clustering than what one would expect at random \cite{Watts1998Collective}.  Thus, clustering has been another graph property of interest for some time.  Here, we are interested in evaluating the extent to which samples exhibit the level of clustering present in the original network.  We employ two notions of clustering, which we now describe.

\subsubsection{Measures}
\label{sec:rep.clustering.measures}
\noindent
\textbf{Local Clustering Coefficient {\sc (CCloc)}.}  The local clustering coefficient \cite{Newman2003Structure} of a node captures the extent to which the node's neighbors are also neighbors of each other.  Formally, the local clustering coefficient of a node is defined as $C_L(v) = \frac{2\ell}{d_v(d_v - 1)}$ where $d_v$ is the degree of node $v$ and $\ell$ is the number of links among the neighbors of $v$.  The average local clustering coefficient for a network is simply $\frac{\sum_{v \in V} C_L(v)}{|V|}$.

\noindent
\textbf{Global Clustering Coefficient {\sc (CCglb)}.}  The global clustering coefficient \cite{Newman2003Structure} is a function of the number of triangles in a network.  It is measured as the number of closed triplets divided by the number of connected triples of nodes.

\subsubsection{Results}
\label{sec:rep.clustering.results}

Results for clustering measures are less consistent than for other measures.  Overall, DFS and RW strategies appear to fare relatively better than others.  We do observe that, for many strategies and networks, estimates of clustering are initially higher-than-actual and then gradually decline (see Figure \ref{fig:rep.clustering}).  This agrees with intuition.  Nodes in clusters should intuitively have more paths leading to them and will, thus, be encountered earlier in a sampling process (as opposed to nodes not embedded in clusters and located in the periphery of a network).  This, then, should be taken into consideration in applications where accurately matching clustering levels is important.

\begin{figure}[htb]
  \centering
  \subfloat[WikiVote ({\sc CCloc})] {\label{fig:rep.ccloc.wikivote}\includegraphics[width=0.2\textwidth]{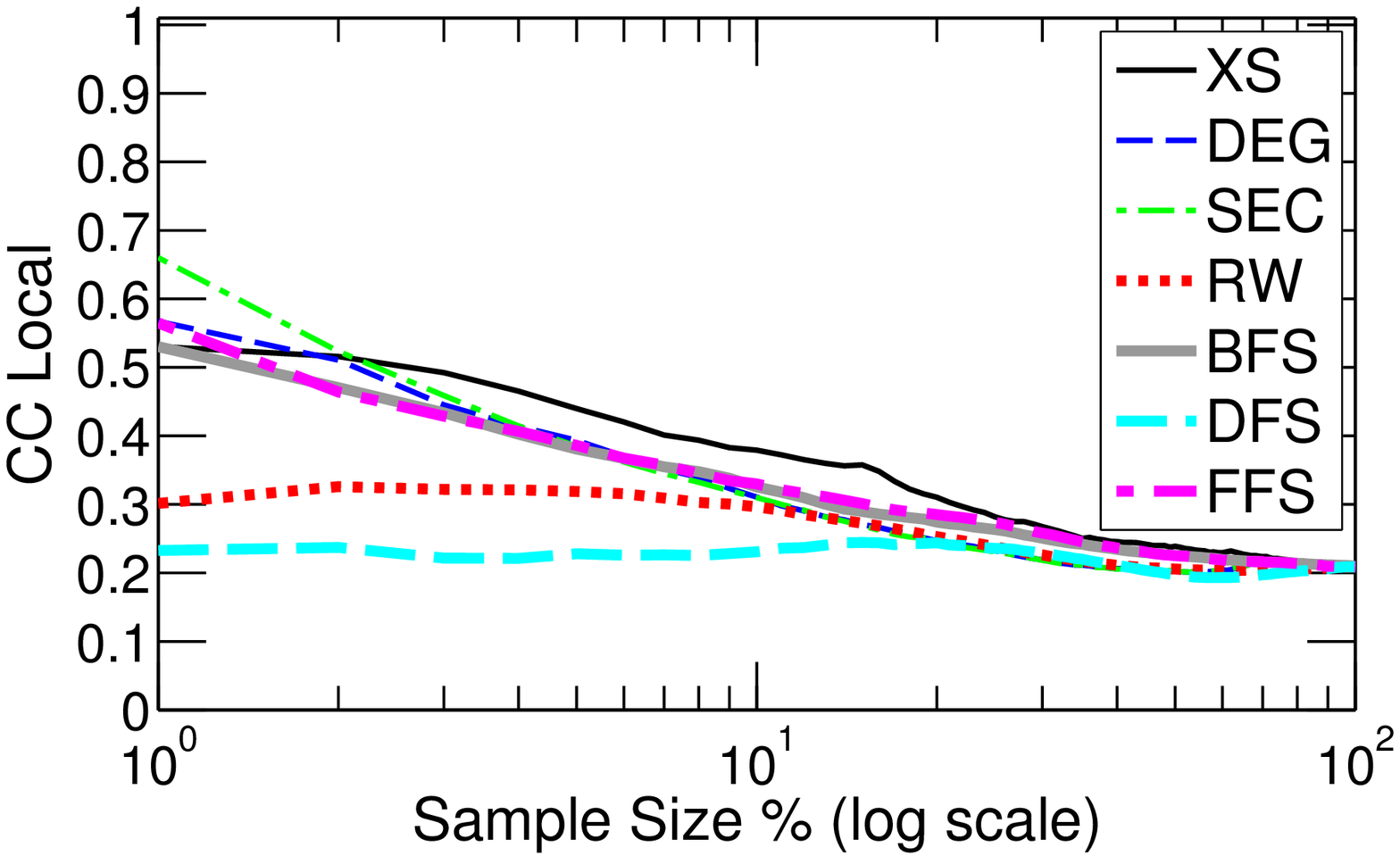}} \vspace{.05cm}
  \subfloat[HEPTh ({\sc CCloc})]{\label{fig:rep.ccloc.hepth}\includegraphics[width=0.2\textwidth]{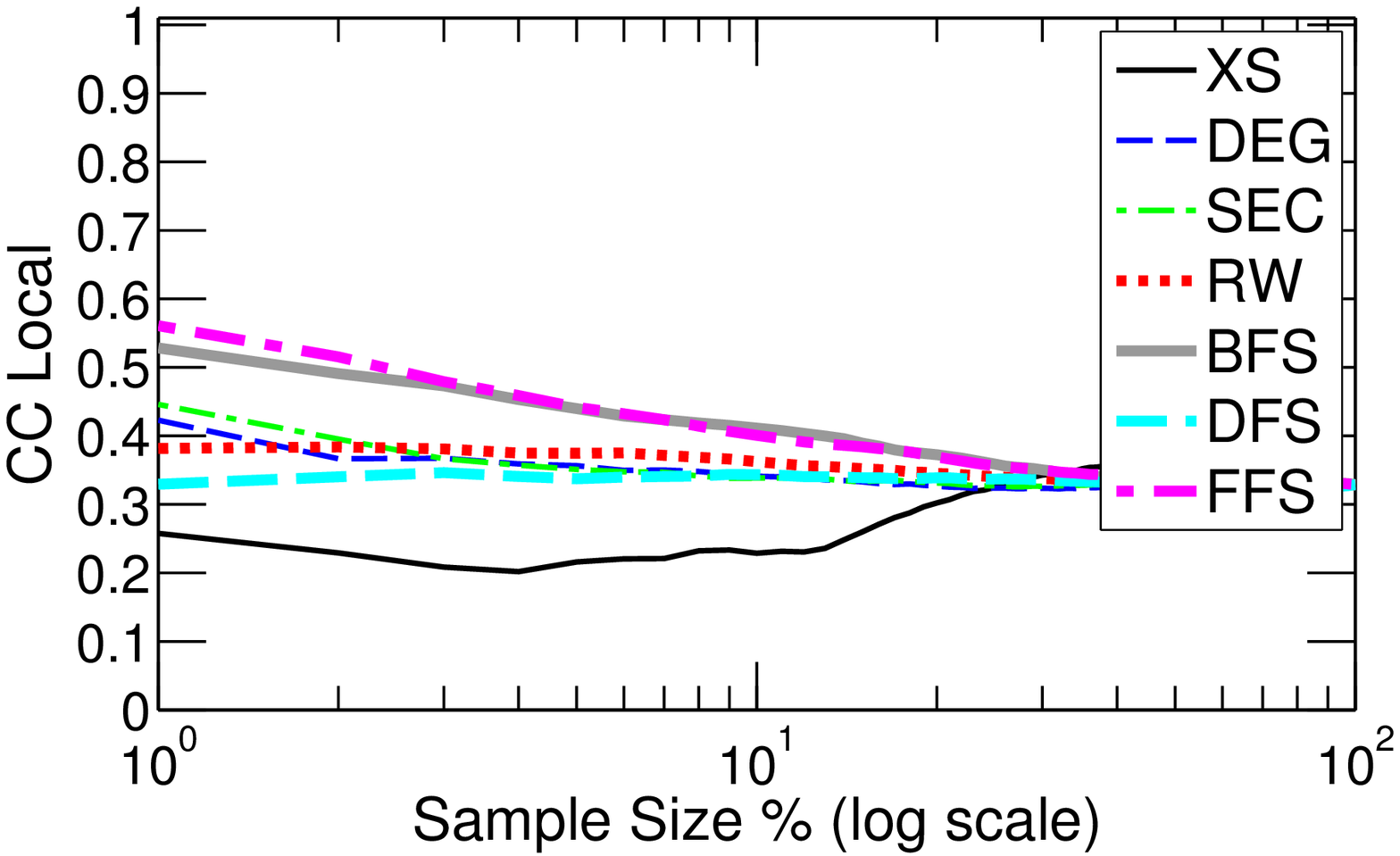}} \\ \vskip -0.01in
  \subfloat[WikiVote ({\sc CCglb})] {\label{fig:rep.ccglb.wikivote}\includegraphics[width=0.2\textwidth]{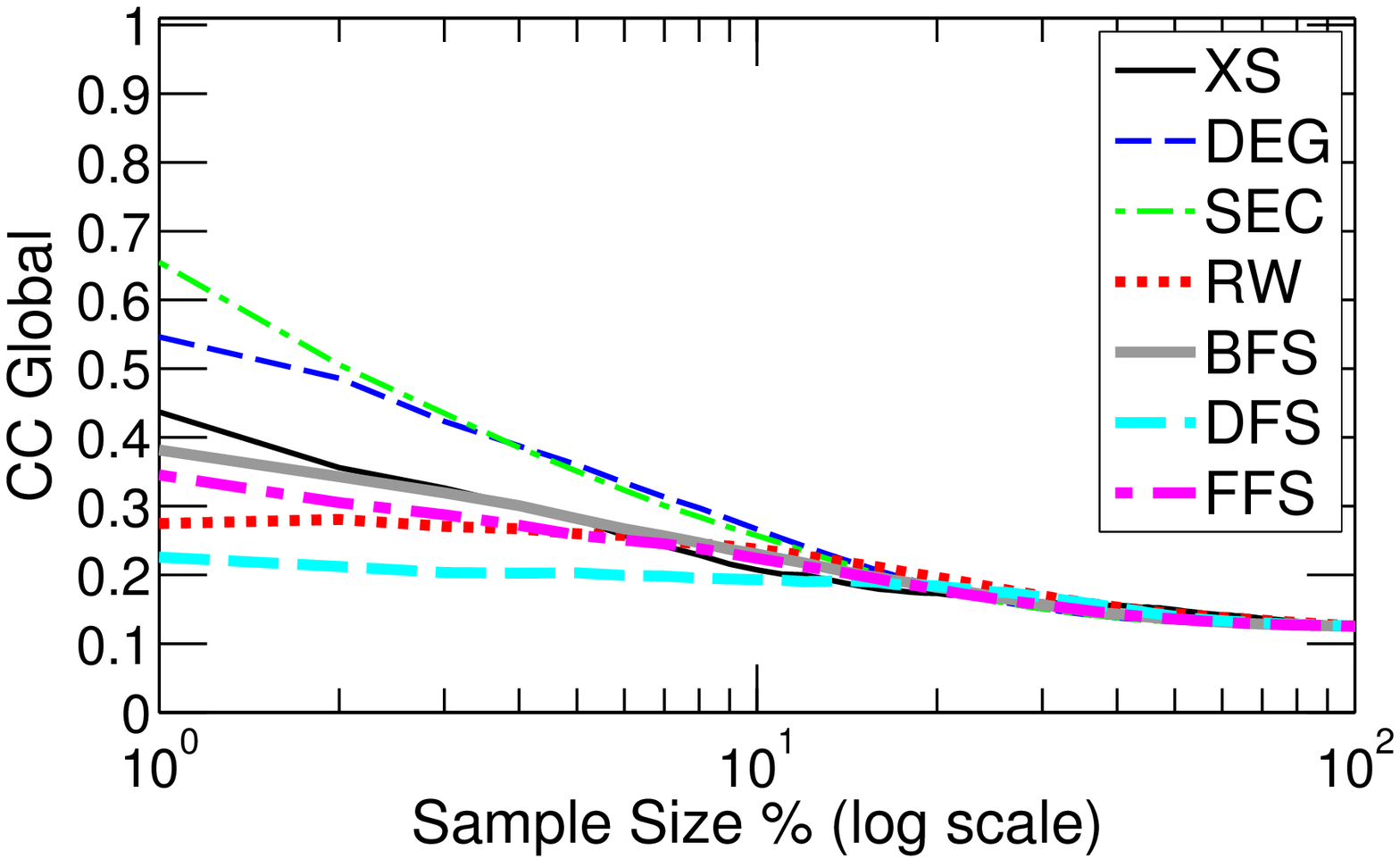}} \vspace{.05cm}
  \subfloat[HEPTh ({\sc CCglb})]{\label{fig:rep.ccglb.hepth}\includegraphics[width=0.2\textwidth]{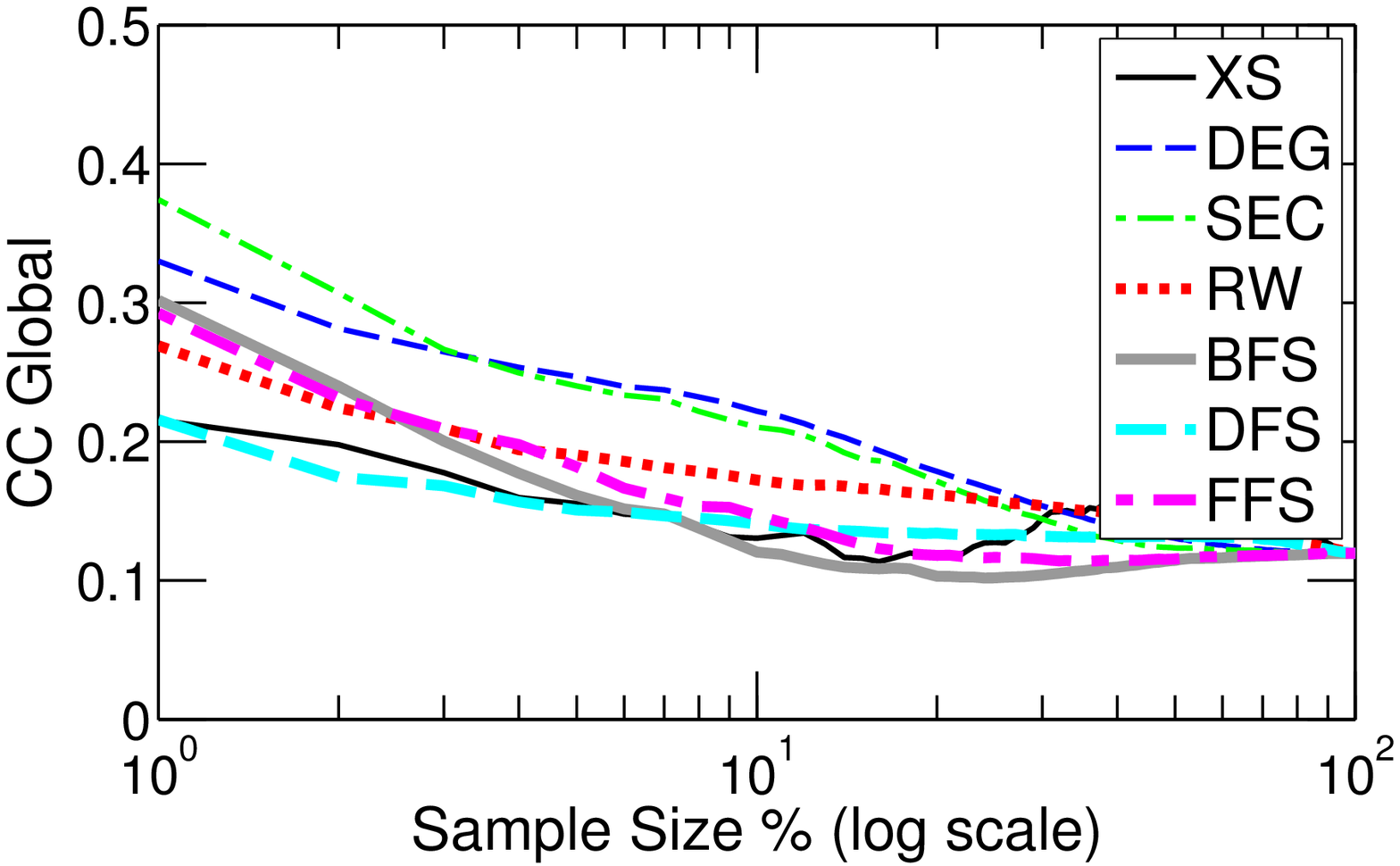}}
\caption{Evaluating {\sc CCglb} and  {\sc CCloc}.}
  \label{fig:rep.clustering}
\vskip -0.15in
\end{figure}

\subsection{Network Reach}
\label{sec:rep.reach}
We propose a new measure of representativeness called \emph{network reach}.  As a newer measure, \emph{network reach} has obviously received considerably less attention than Degree and Clustering within the existing literature, but it is, nevertheless, a vital measure for a number of important applications (as we will see in Section \ref{sec:applications}).  \emph{Network reach} captures the extent to which a sample \emph{covers} a network.  Intuitively, for a sample to be truly representative of a large network, it should consist of nodes from diverse portions of the network, as opposed to being relegated to a small ``corner'' of the graph.  This concept will be made more concrete by discussing in detail the two measures of \emph{network reach} we employ:  \emph{community reach} and the \emph{discovery quotient}.

\subsubsection{Measures}
\label{sec:rep.reach.measures}
\noindent
\textbf{Community Reach ({\sc CNM} and {\sc RAK}).}  Many real-world networks exhibit what is known as \emph{community structure}.  A \emph{community} can be loosely defined as a set of nodes more densely connected among themselves than to other nodes in the network.   Although there are many ways to represent community structure depending on various factors such as whether or not overlapping is allowed, in this work, we represent community structure as a \emph{partition}: a collection of disjoint subsets whose union is the vertex set $V$ \cite{Fortunato2010Community}.  Under this representation, each subset in the partition represents a community. The task of a community detection algorithm is to identify a partition such that vertices within the same subset in the partition are more densely connected to each other than to vertices in other subsets \cite{Fortunato2010Community}.  For the criterion of \emph{community reach}, a sample is more representative of the network if it consists of nodes from more of the communities in the network.  We measure \emph{community reach} by taking the number of communities represented in the sample and dividing by the total number of communities present in the original network.  Since a community is essentially a cluster of nodes, one might wonder why we have included \emph{community reach} as a measure of \emph{network reach}, rather than as a measure of \emph{clustering}.  The reason is that we are slightly less interested in the structural details of communities detected here.  Rather, our aim is to assess how ``spread out'' a sample is across the network.  Since community detection is somewhat of an inexact science (e.g. see \cite{Good2010Performance}), we measure \emph{community reach} with respect to two separate algorithms.  We employ both the method proposed by Clauset et al. in \cite{Clauset2004Finding} (denoted as {\sc CNM}) and the approach proposed by Raghavan et al. in \cite{Raghavan2007Near} (denoted as {\sc RAK}). Essentially, for our purposes, we are defining communities simply as the output of a community detection algorithm.

\noindent
\textbf{Discovery Quotient {\sc (DQ)}.}   An alternative view of \emph{network reach} is to measure the proportion of the network that is \emph{discovered} by a sampling strategy.  The number of nodes discovered by a strategy is defined as $|S \cup N(S)|$.  The \emph{discovery quotient} is this value normalized by the total number of nodes in a network:  $\frac{|S \cup N(S)|}{|V|}$.  Intuitively, we are defining the \emph{reach} of a sample here by measuring the extent to which it is one hop away from the rest of the network.  As we will discuss in Section \ref{sec:applications}, samples with high \emph{discovery quotients} have several important applications.  Note that a simple greedy algorithm for coverage problems such as this has a well-known sharp approximation bound of $1-1/e$ \cite{Maiya2010Expansion,Feige1998Threshold}.  However, link-trace sampling is restricted to selecting subsequent sample elements from the current neighborhood $N(S)$ at each iteration, which results in a much smaller search space.  Thus, this approximation guarantee can be shown not to hold within the context of link-trace sampling.

\subsubsection{Results}
\label{sec:rep.reach.results}
As shown in Figure \ref{fig:rep.reach}, the XS strategy displays the overwhelmingly best performance on all three measures of \emph{network reach}.  We highlight several observations here.  First, the extent to which the XS strategy outperforms all others on the {\sc RAK} and {\sc CNM} measures is quite striking.  We posit that the expansion bias of the XS strategy ``pushes'' the sampling process towards the inclusion of new communities not already seen (see also \cite{Maiya2010Sampling}).  In Section \ref{sec:biases.xs}, we will analytically examine this connection between expansion bias and \emph{community reach}.  On the other hand, the SEC method appears to be among the least effective in reaching different communities or clusters.  We attribute this to the fact that SEC preferentially selects nodes with many connections to nodes already sampled.  Such nodes are likely to be members of clusters already represented in the sample.  Second, on the {\sc DQ} measure, it is surprising that the DS strategy, which explicitly selects high degree nodes, often fails to even come close to the XS strategy.  We partly attribute this to an overlap in the neighborhoods of well-connected nodes.  By explicitly selecting nodes that contribute to \emph{expansion}, the XS strategy is able to discover a much larger proportion of the network in the same number of steps - in some cases, by actively sampling comparatively \emph{lower} degree nodes.  Finally, it is also surprising that the BFS strategy, widely used to crawl and explore online social networks (e.g \cite{Mislove2007Measurement}) and other graphs (e.g. \cite{Najork2001Breadthfirst}), performs quite dismally on all three measures.  In short, we find that nodes contributing most to the expansion of the sample are unique in that they provide specific and significant advantages over and above those provided by nodes that are simply well-connected and those accumulated through standard BFS-based crawls.  These and previously mentioned results are in contrast to the conventional wisdom followed in much of the existing literature (e.g. \cite{Najork2001Breadthfirst,Mislove2007Measurement,Adamic2001Search,Cohen2003Efficient,Kurant2010Bias}).

\begin{figure}[htb]
  \centering
  \subfloat[HEPTh {\sc (RAK)} ] {\label{fig:rep.hepth.commlp}\includegraphics[width=0.2\textwidth]{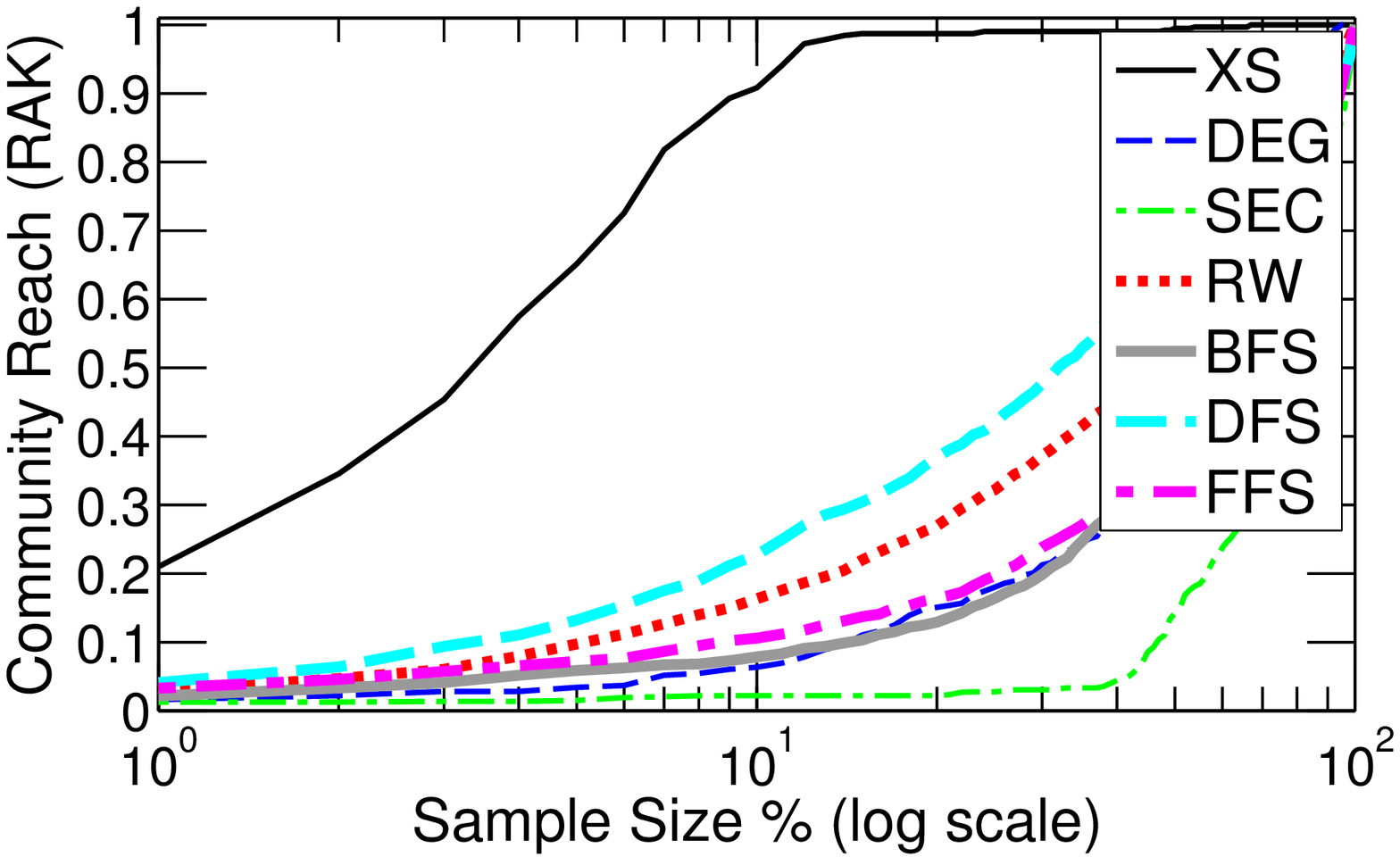}} \vspace{.05cm}
  \subfloat[Amazon {\sc (RAK)}]{\label{fig:rep.amazon.commlp}\includegraphics[width=0.2\textwidth]{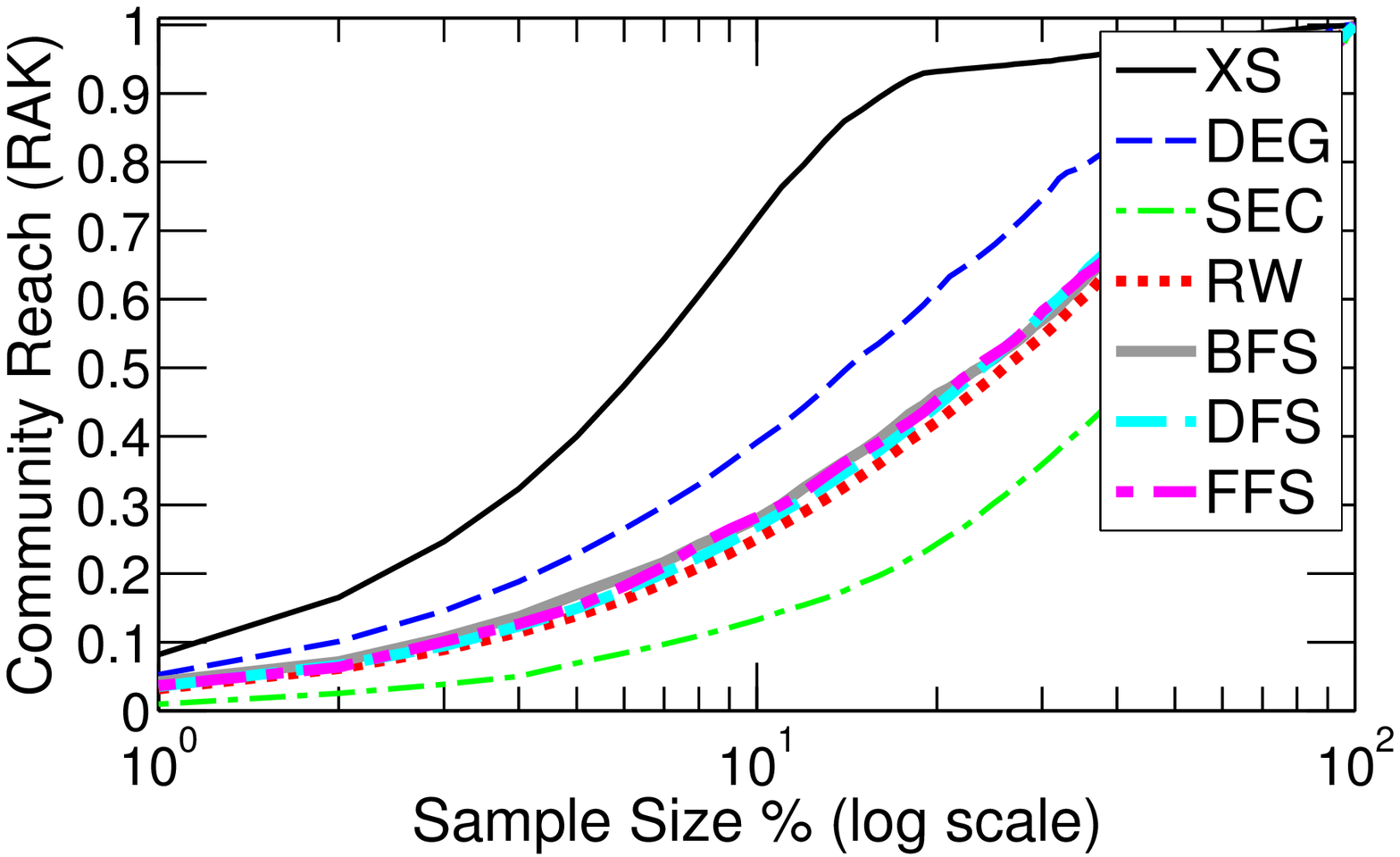}} \\ \vskip -0.01in
  \subfloat[HEPTh {\sc (CNM)}] {\label{fig:rep.hepth.commfg}\includegraphics[width=0.2\textwidth]{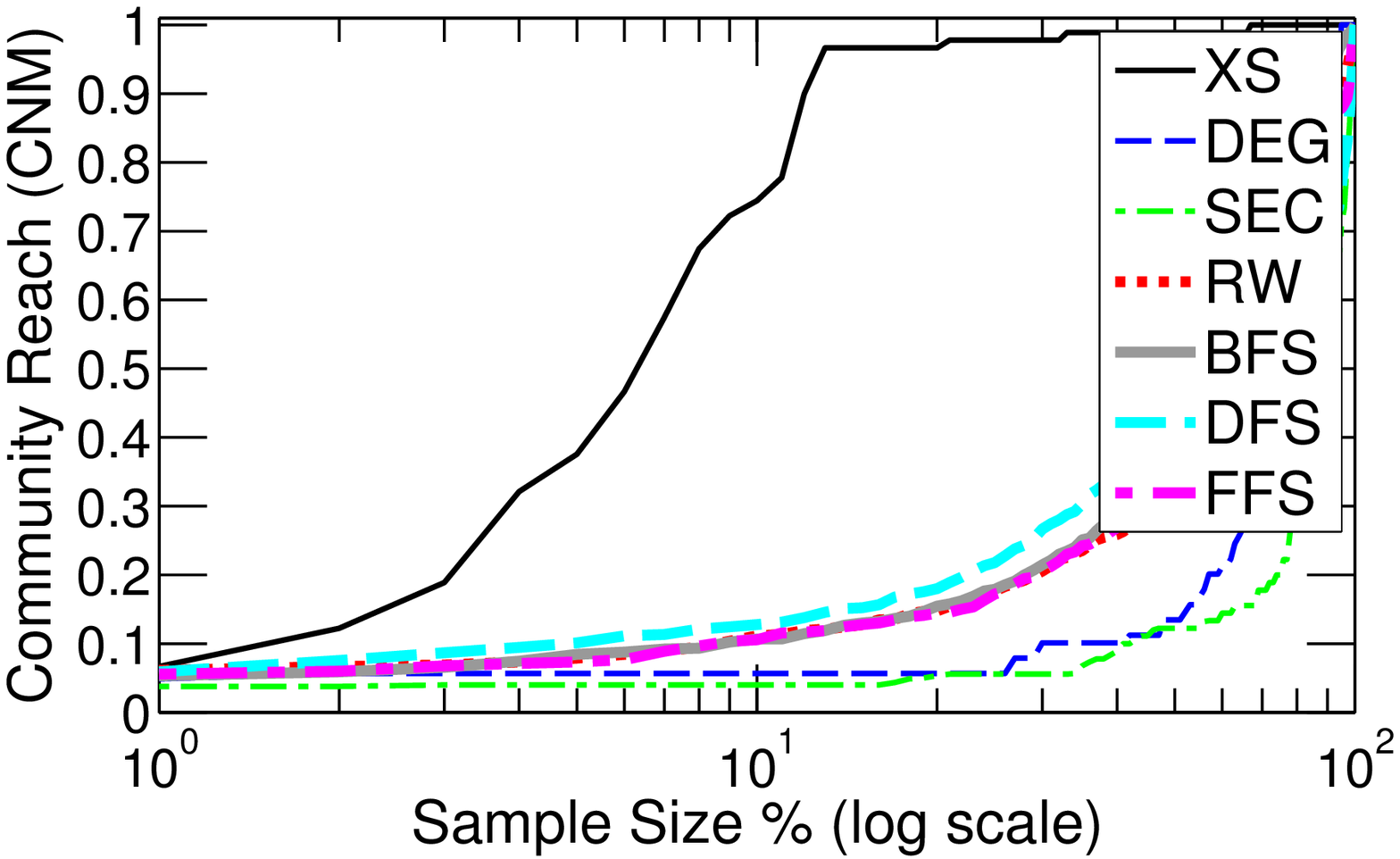}} \vspace{.05cm}
  \subfloat[Amazon {\sc (CNM)} ]{\label{fig:rep.amazon.commfg}\includegraphics[width=0.2\textwidth]{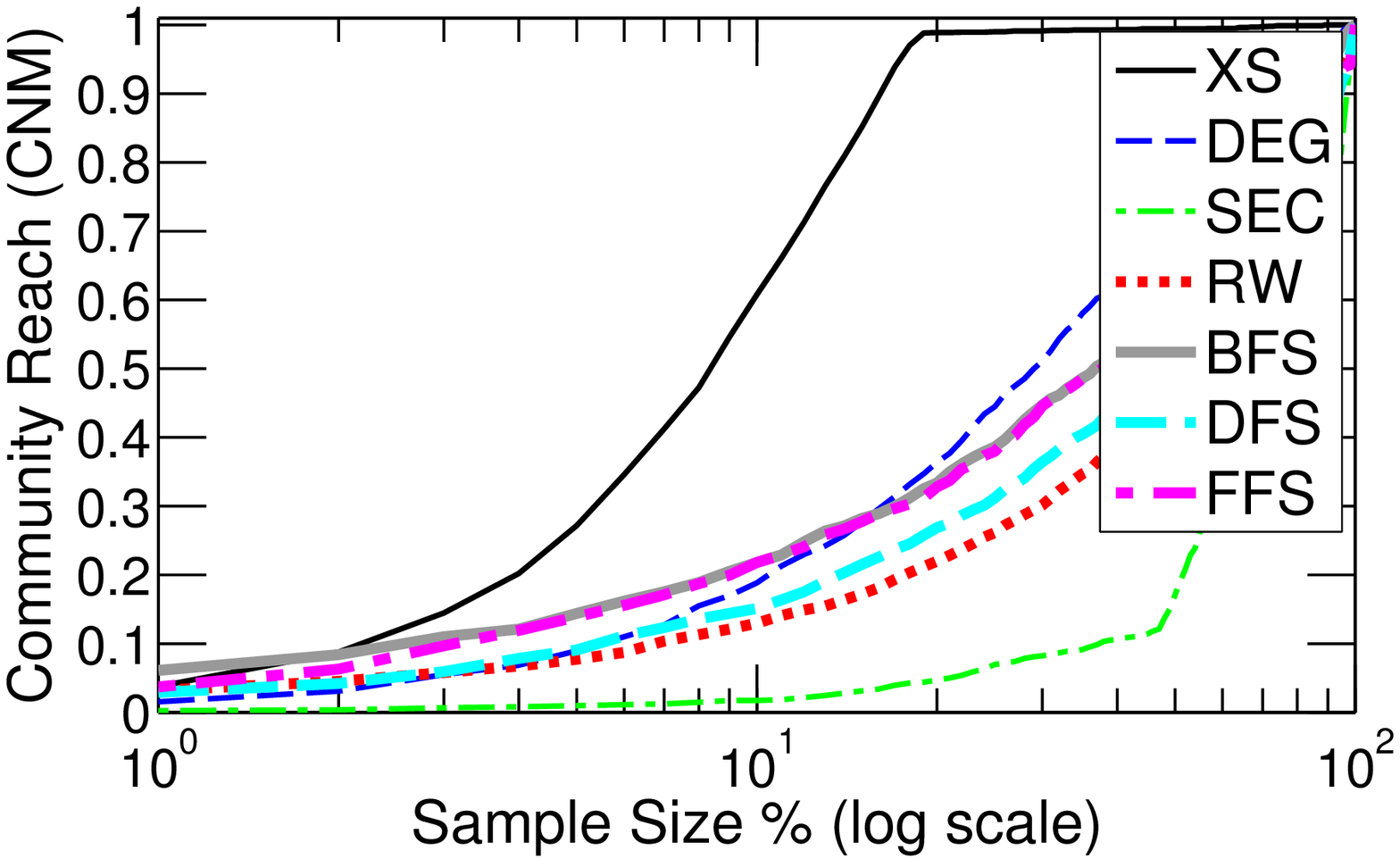}} \\ \vskip -0.01in
  \subfloat[HEPTh {\sc (DQ)}] {\label{fig:rep.hepth.dq}\includegraphics[width=0.2\textwidth]{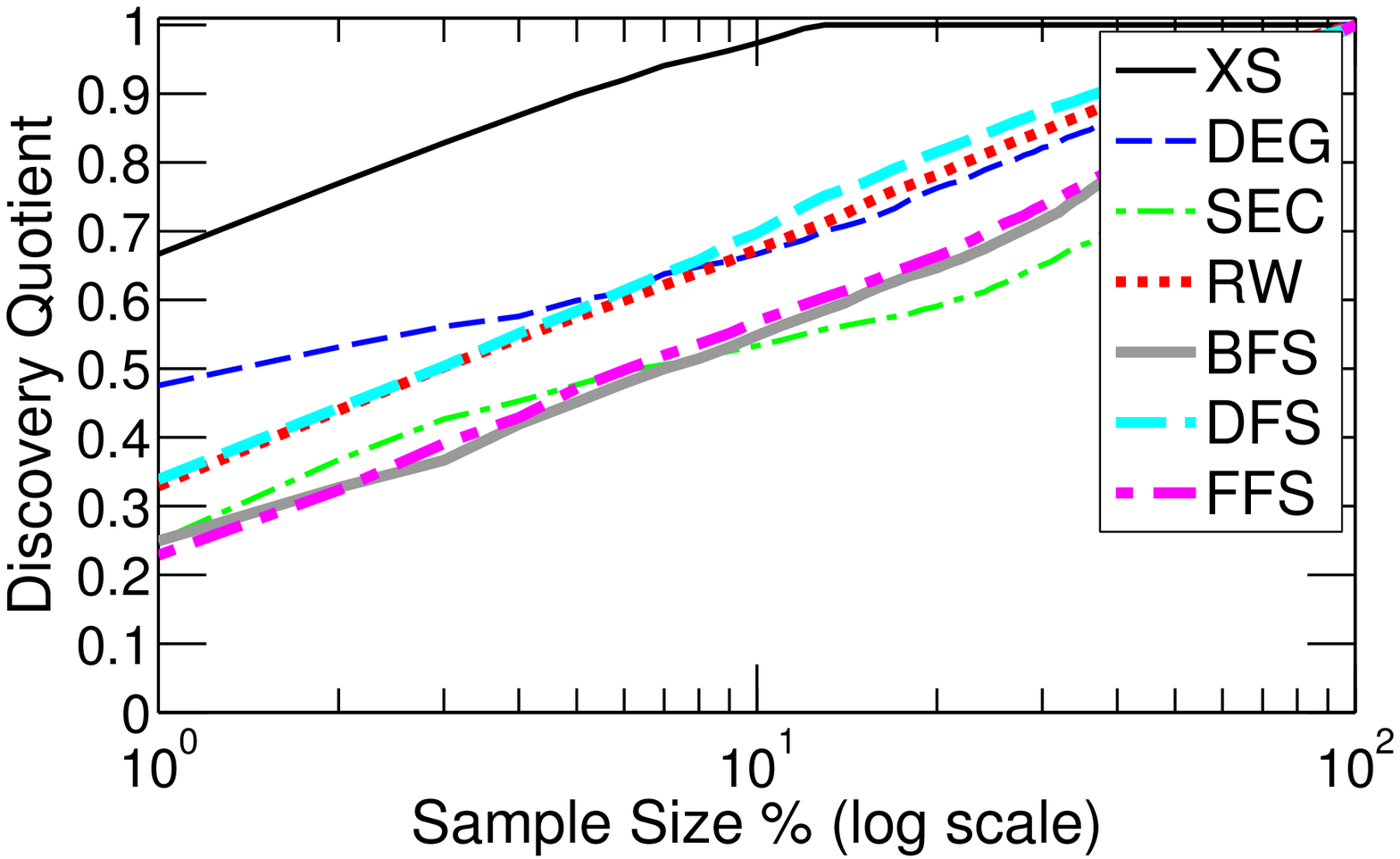}} \vspace{.05cm}
  \subfloat[Amazon {\sc (DQ)}]{\label{fig:rep.amazon.dq}\includegraphics[width=0.2\textwidth]{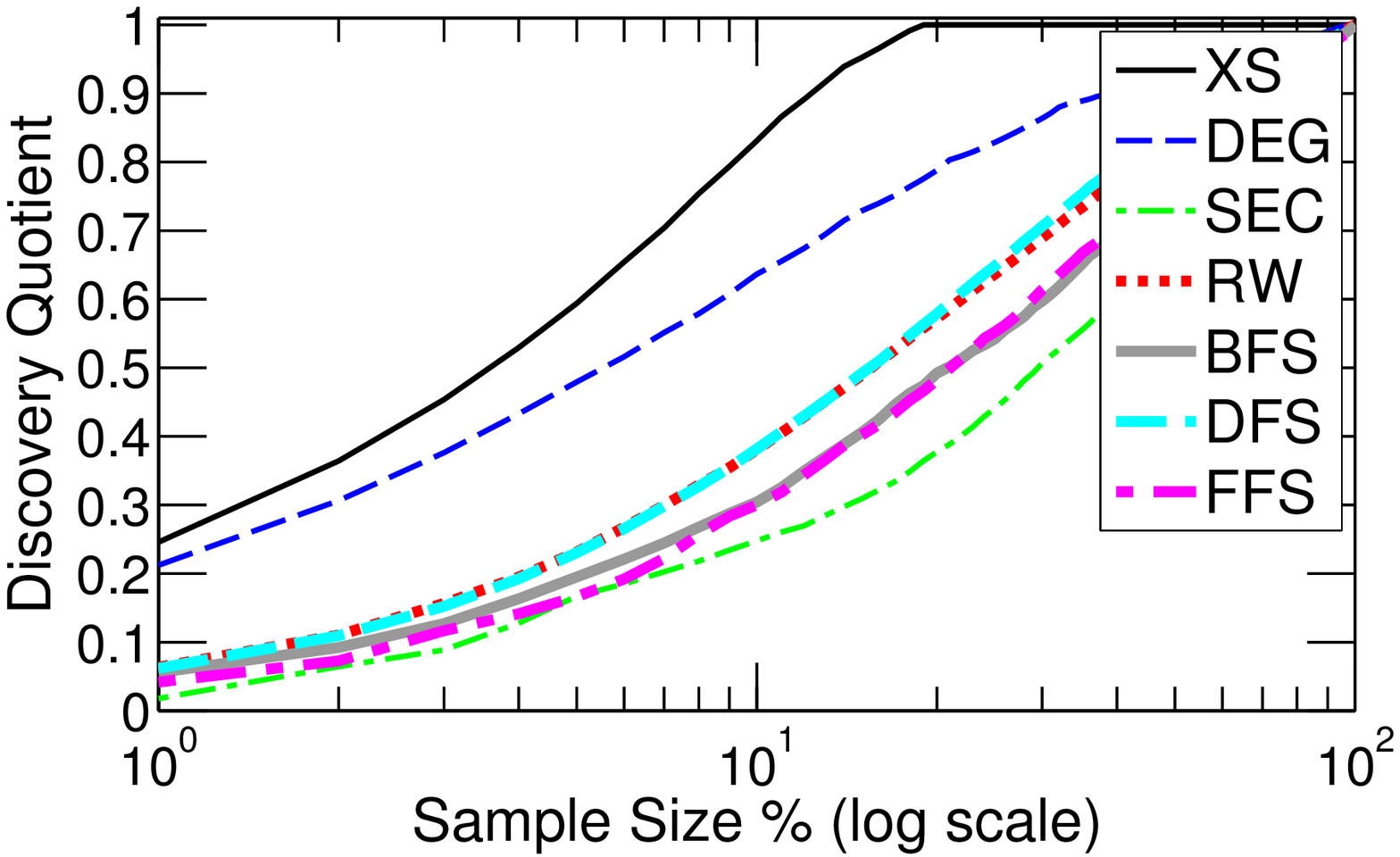}} 
\caption{ {\footnotesize Evaluating \emph{network reach}.  Results for remaining networks are similar with XS exhibiting superior performance on all three criteria.} }
  \label{fig:rep.reach}
\vskip -0.15in
\end{figure}

\subsection{A Note on Seed Sensitivity}
\label{sec:rep.seedsensitivity}
As described, link-trace sampling methods are initiated from randomly selected seeds. This begs the question:  How sensitive are these results to the seed supplied to a strategy?  Figure \ref{fig:std} shows the standard deviation of each sampling strategy for both \emph{hub inclusion} and \emph{network reach} as sample size grows.  We generally find that methods with the most explicit biases (XS, SEC, DS) tend to exhibit the least seed sensitivity and variability, while the remaining methods (BFS, DFS, FFS, RW) exhibit the most.  This trend is exhibited across all measures and all datasets.

\begin{figure}[htb]
  \centering
  \subfloat[Epinions ({\sc DQ})] {\label{fig:dq_std.epinions}\includegraphics[width=0.2\textwidth]{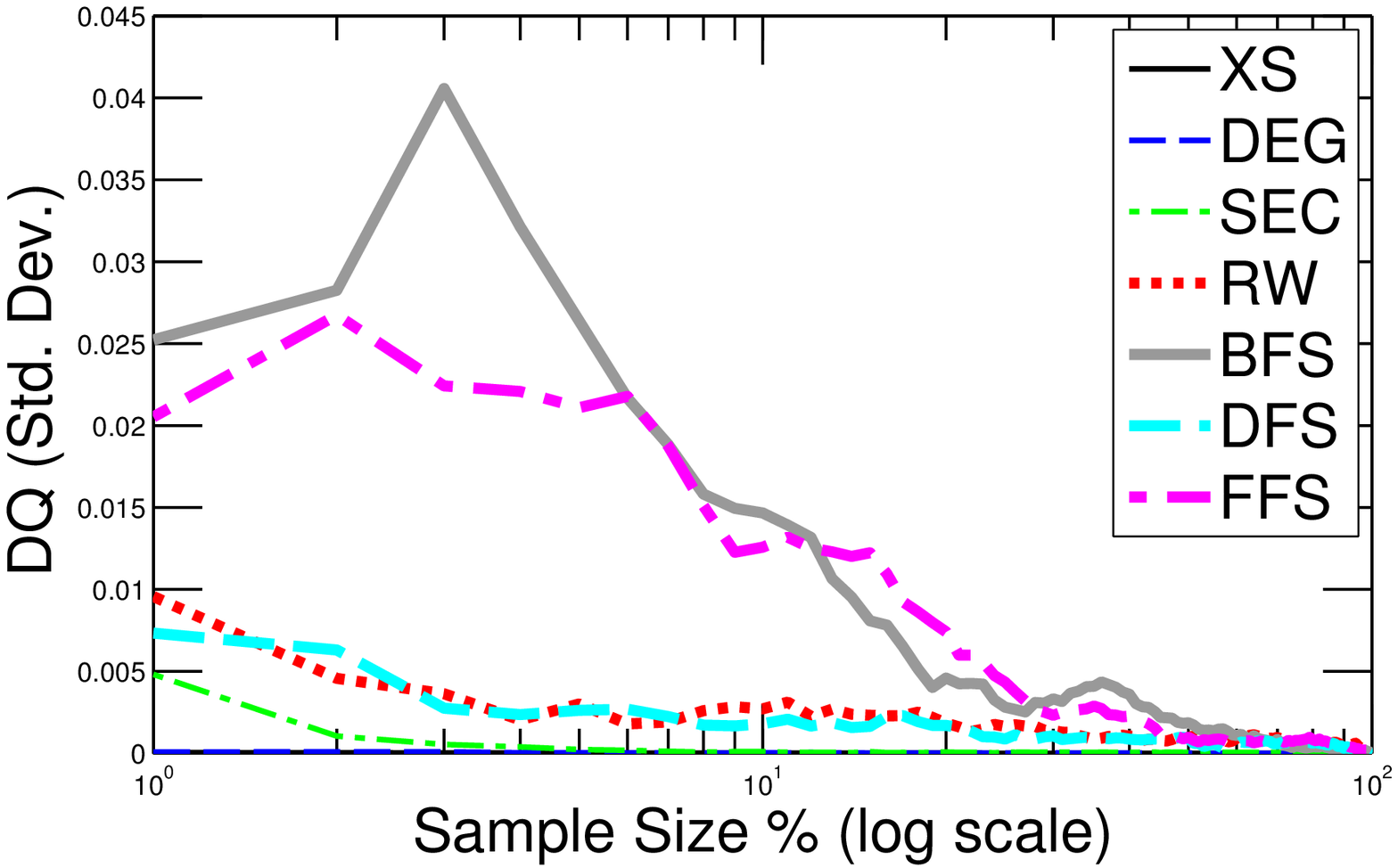}} \vspace{.05cm}
  \subfloat[Epinions ({\sc Hubs})]{\label{fig:std.hubs_std.epinions}\includegraphics[width=0.2\textwidth]{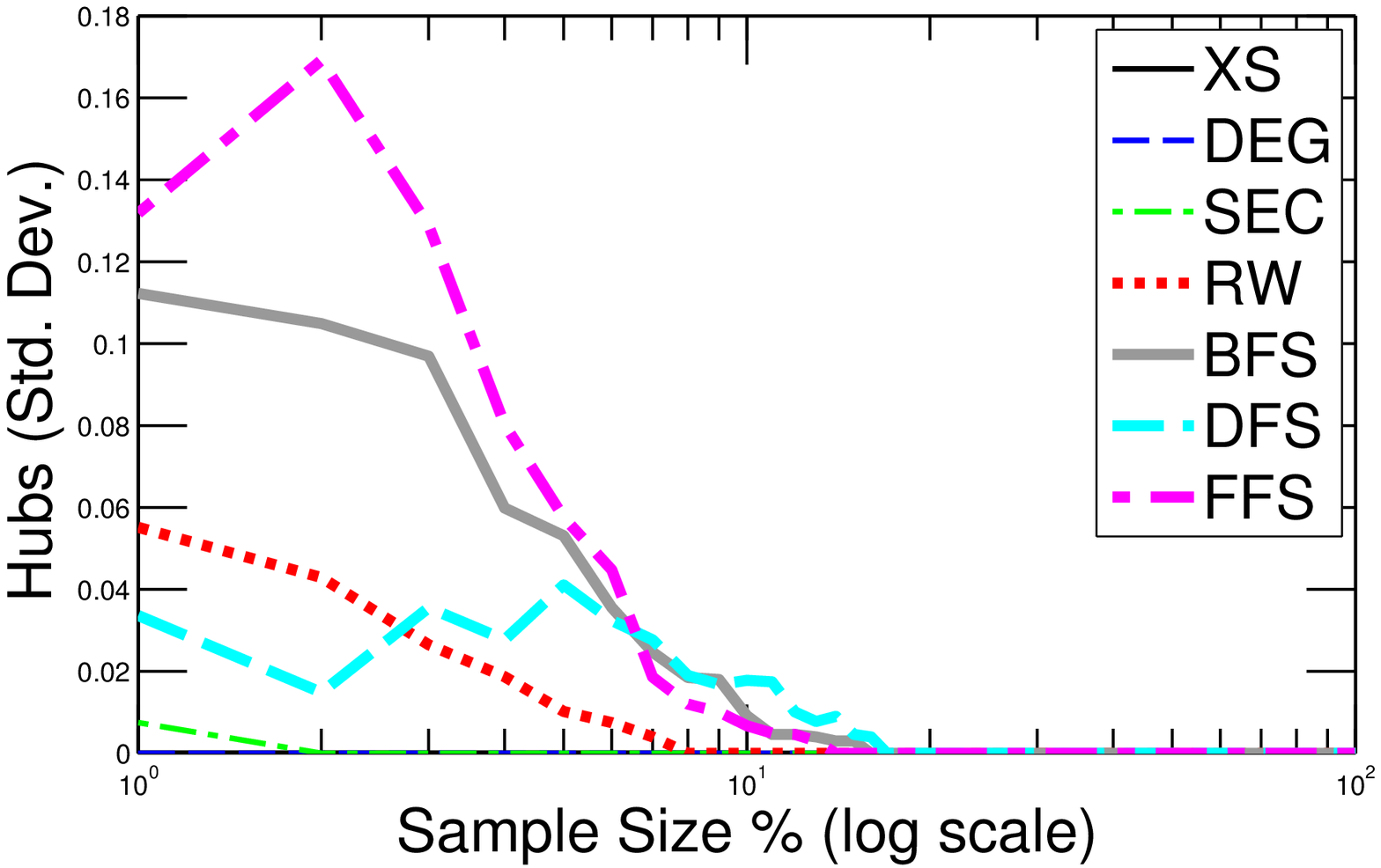}}
\caption{Standard deviation for {\sc DQ} and {\sc Hubs} on Epinions network.  Results are similar for remaining networks.}
  \label{fig:std}
\end{figure}

\section{Analyzing Sampling Biases}
\label{sec:biases}
Let us briefly summarize two main observations from Section \ref{sec:rep}.  We saw that the XS strategy dramatically outperformed all others in accumulating nodes from many different communities.  We also saw that the SEC strategy was often a reasonably good approximation to directly sampling high degree nodes and locates the set of most well-connected nodes significantly faster than most other methods.  Here, we turn our attention to analytically examining these observed connections.  We begin by briefly summarizing some existing analytical results.

\subsection{Existing Analytical Results}
\label{sec:biases.existing}
\noindent
\textbf{Random Walks (RW).}  There is a fairly large body of research on random walks and Markov chains (see \cite{Lovasz1994Random} for an excellent survey).  A well-known analytical result states that the probability (or \emph{stationary} probability) of residing at any node $v$ during a random walk on a connected, undirected graph converges with time to $\frac{d_{v}}{2\cdot|E|}$, where $d_{v}$ is the degree of node $v$ \cite{Lovasz1994Random}.  In fact, the \emph{hitting time} of a random walk (i.e. the expected number of steps required to reach a node beginning from any node) has been analytically shown to be directly related to this stationary probability \cite{Hopcroft2007Manipulationresistant}.  Random walks, then, are naturally biased towards high degree (and high PageRank) nodes, which provides some theoretical explanation as to why RW performs slightly better than other strategies (e.g. BFS) on measures such as \emph{hub inclusion}.  However, as shown in Figure \ref{fig:rep.degree}, it is nowhere near the best performers.  Thus, these analytical results appear only to hold in the limit and fail to predict actual sampling performance.

\noindent
\textbf{Degree Sampling (DS).}  In studying the problem of searching peer-to-peer networks, Adamic et al. \cite{Adamic2001Search} proposed and analyzed a greedy search strategy very similar to the DS sampling method. This strategy, which we refer to as a degree-based walk, was analytically shown to quickly find the highest-degree nodes and quickly cover large portions of scale-free networks.  Thus, these results provide a theoretical explanation for performance of the DS strategy on measures such as \emph{hub inclusion} and the \emph{discovery quotient}.

\noindent
\textbf{Other Results.}  As mentioned in Section \ref{sec:relatedwork}, to the best of our knowledge, much of the other analytical results on sampling bias focus on \emph{negative} results \cite{Lakhina2003Sampling,Costenbader2003Stability,Kurant2010Bias,Achlioptas2005Bias,Stumpf2005Subnets}.  Thus, these works, although intriguing, may not provide much help in the way of explaining \emph{positive} results shown in Section \ref{sec:rep}.  

~\\
We now analyze two methods for which there are little or no existing analytical results:  XS and SEC.

\subsection{Analyzing XS Bias}
\label{sec:biases.xs}
A widely used measure for the ``goodness'' or the strength of a community in graph clustering and community detection is \emph{conductance} \cite{kannan2004clusterings}, which is a function of the fraction of total edges emanating from a sample (lower values mean stronger communities):  $$\varphi(S) = \frac{\sum_{i \in S,j \in \overline{S}} a_{ij}} {\min(a(S), a(\overline{S}))}$$ where $a_{ij}$ are entries of the adjacency matrix representing the graph and $a(S) = \sum_{i \in S} \sum_{j \in V}a_{ij}$, which is the total number of edges incident to the node set $S$.  

It can be shown that, provided the conductance of communities is sufficiently low, sample expansion is directly affected by community structure.  Consider a  simple random graph model with vertex set $V$ and a community structure represented by partition $C=\{C_{1},\ldots,C_{|C|}\}$ where $C_{1} \cup \ldots \cup C_{|C|}=V$.  Let $e_{in}$ and $e_{out}$ be the number of each node's edges pointing within and outside the node's community, respectively.  These edges are connected uniformly at random to nodes either within or outside a node's community, similar to a configuration model (e.g., \cite{Chung2002Connected}).  Note that both $e_{in}$ and $e_{out}$ are related directly to conductance.  When conductance is lower, $e_{out}$ is smaller\footnote{Suppose conductance of a vertex set $Y$ is $\varphi(Y)$, the total number of edges incident to $Y$ is $e$, and $e_{in}$ and $e_{out}$ are random variables denoting the inward and outward edges, respectively, of each node (as opposed to constant values).  Then, $\mathbb{E}(e_{out}) = \frac{e\varphi(Y)}{|Y|}$ and $\mathbb{E}(e_{in})= \frac{2e(1-\varphi(Y))}{|Y|}$.  If $\varphi(Y) < \frac{2}{3}$, then $\mathbb{E}(e_{out}) < \mathbb{E}(e_{in})$.  {\scriptsize (In this example, the expectations are over nodes in $Y$ only.)}}  as compared to $e_{in}$.  The following theorem expresses the link between expansion and \emph{community reach} in terms of these inward and outward edges.  

\begin{thm}
\label{thm:xsbias}
Let $S$ be the current sample, $v$ be a new node to be added to $S$, and $n$ be the size of $v$'s community.  If $e_{out} < \frac{|V|(e_{in})^2}{n(|V|+e_{in}|S|)}$, then the expected expansion of $S \cup \{v\}$ is higher when $v$ is in a new community than when $v$ is in a current community.
\end{thm}
\begin{proof}
Let $X_{new}$ be the expected value for $|N(\{v\})-N(S) \cup S|$ when $v$ is in a new community and let $X_{curr}$ be the expected value when not.  We compute an upper bound on $X_{curr}$ and a lower bound on $X_{new}$. 

~\\
\noindent
Deriving $X_{curr}$:  Assume $v$ is affiliated with a current community already represented by at least one node in $S$.  Since we are computing an upper bound on $X_{curr}$, we assume there is exactly one node from $S$ within $v$'s community, as this is the minimum for $v$'s community to be a \emph{current} community. By the linearity of expectations, the upper bound on $X_{curr}$ is $e_{out} + \frac{(n-e_{in})e_{in}}{n}$, where the term $\frac{(n-e_{in})e_{in}}{n}$ is the expected number of nodes in $v$'s community that are both linked to $v$ \emph{and} in the set $V-(N(S) \cup S)$.  

~\\
\noindent
Deriving $X_{new}$:  Assume $v$ belongs to a new community not already represented in $S$.  (By definition, no nodes in $S$ will be in $v$'s community.)  Applying the linearity of expectations once again, the lower bound on $X_{new}$ is $e_{in} - e_{out}|S|\frac{e_{in}}{|V|}$, where the term $e_{out}|S|\frac{e_{in}}{|V|}$ is the expected number of nodes in $v$'s community that are both linked to $v$ \emph{and} already in $N(S)$. 

~\\
Solving for $e_{out}$, if $e_{out} < \frac{|V|(e_{in})^2}{n(|V|+e_{in}|S|)}$, then $X_{new} > X_{curr}$. 
\end{proof}

Theorem \ref{thm:xsbias} shows analytically the link between expansion and community structure - a connection that, until now, has only been empirically demonstrated \cite{Maiya2010Sampling}. Thus, a theoretical basis for performance of the XS strategy on \emph{community reach} is revealed.

\subsection{Analyzing SEC Bias}
\label{sec:biases.sec}
Recall that the SEC method uses the degree of a node $v$ in the induced subgraph $G_{S \cup \{v\}}$ as an estimation for the degree of $v$ in $G$.  In Section \ref{sec:rep}, we saw that this choice performs quite well in practice.  Here, we provide theoretical justification for the SEC heuristic.  Consider a random network $G$ with some arbitrary expected degree sequence (e.g. a power law random graph under the so-called $G(\mathbf{w})$ model \cite{Chung2002Connected}) and a sample $S \subset V$.  Let $d(\cdot, \cdot)$ be a function that returns the expected degree of a given node in a given random network (see \cite{Chung2002Connected} for more information on \emph{expected} degree sequences).  Then, it is fairly straightforward to show the following holds.

\begin{prop}
\label{prop:secbias}
For any two nodes $v,w \in N(S)$, \\if $d(v, G) \geq d(w, G)$, then $d(v, G_{S \cup \{v\}}) \geq d(w, G_{S \cup \{w\}})$.
\end{prop}
\begin{proof}
The probability of an edge between any two nodes $i$ and $j$ in G is $\frac{d(i, G) \cdot d(j, G)}{\Delta}$ where $\Delta = \sum_{m \in V}d(m,G)$.  Let $\delta = d(v, G_{S \cup \{v\}}) - d(w, G_{S \cup \{w\}})$.  Then, 
\begin{align}
 \delta &= \sum_{x \in S}\frac{d(x, G) \cdot d(v,G)}{\Delta}  - \sum_{x \in S}\frac{d(x, G) \cdot d(w,G)}{\Delta}  \\
					    ~ &=  (d(v,G) - d(w,G))  \sum_{x \in S}\frac{d(x, G)}{\Delta} 
\end{align}
Since $\delta \geq 0$ only when $d(v,G) \geq d(w,G)$, the proposition holds.
\end{proof}

Combining Proposition \ref{prop:secbias} with analytical results from \cite{Adamic2001Search} (described in Section \ref{sec:biases.existing}) provides a theoretical basis for observed performance of the SEC strategy on measures such as \emph{hub inclusion}.  Finally, recall from Section \ref{sec:rep.degree.results} that the extent to which SEC matched the performance of DS on {\sc Hubs} seemed to partly depend on the tail of degree distributions.  Proposition \ref{prop:secbias} also yields insights into this phenomenon.  Longer and denser tails allow for more ``slack'' when deviating from these expectations of random variables (as in real-world link patterns that are not purely random).

\section{Applications for Our Findings}
\label{sec:applications}
We now briefly describe ways in which some of our findings may be exploited in important, real-world applications.  Although numerous potential applications exist, we focus here on three areas:  1) Outbreak Detection 2) Landmarks and Graph Exploration 3) Marketing.

\subsection{Practical Outbreak Detection}
What is the most effective and efficient way to predict and prevent a disease outbreak in a social network?  In a recent paper, Christakis and Fowler studied outbreak detection of the H1N1 flu among college students at Harvard University \cite{Christakis2010Social}.  Previous research has shown that well-connected (i.e. high degree) people in a network catch infectious diseases earlier than those with fewer connections \cite{Cohen2003Efficient,Zuckerman2001What,Feld1991Why}.  Thus, \emph{monitoring} these individuals allows forecasting the progression of the disease (a boon to public health officials) and  \emph{immunizing} these well-connected individuals (when immunization is possible) can prevent or slow further spread.  Unfortunately, identifying well-connected individuals in a population is non-trivial, as access to their friendships and connections is typically not fully available.  And, collecting this information is time-consuming, prohibitively expensive, and often impossible for large networks.  Matters are made worse when realizing that most existing network-based techniques for immunization selection and outbreak detection assume full knowledge of the global network structure (e.g. \cite{Tong2010Vulnerability,Leskovec2007Costeffective}).  This, then, presents a prime opportunity to exploit the power of \emph{sampling}.

To identify well-connected students and predict the outbreak, Christakis and Fowler \cite{Christakis2010Social} employed a sampling technique called \emph{acquaintance sampling} (ACQ) based on the so-called friendship paradox \cite{Zuckerman2001What,Cohen2003Efficient,Christakis2010Social}.  The idea is that random neighbors of randomly selected nodes in a network will tend to be highly-connected \cite{Cohen2003Efficient,Zuckerman2001What,Feld1991Why}.  Christakis and Fowler \cite{Christakis2010Social}, therefore, sampled random friends of randomly selected students with the objective of constructing a sample of highly-connected individuals.  Based on our aforementioned results, we ask:  Can we do better than this ACQ strategy?  In previous sections, we showed empirically and analytically that the SEC method performs exceedingly well in accumulating hubs.  (It also happens to require less information than DS and XS, the other top performers.)  Figure \ref{fig:outdet} shows the sample size required to locate the top-ranked well-connected individuals for both SEC and ACQ.  The performance differential is quite remarkable, with the SEC method faring overwhelmingly better in quickly zeroing in on the set of most well-connected nodes.  Aside from its superior performance, SEC has one additional advantage over the ACQ method employed by Christakis and Fowler.  The ACQ method assumes that nodes in $V$ can be selected uniformly at random.  It is, in fact, dependent on this \cite{Cohen2003Efficient}.  (ACQ, then, is \emph{not} a link-trace sampling method.)  By contrast, SEC, as a pure link-trace sampling strategy, has no such requirement and, thus, can be applied in realistic scenarios for which ACQ is unworkable.
\begin{figure}[ht]
  \centering
  \subfloat[Slashdot] {\label{fig:outdet.slashdot}\includegraphics[width=0.2\textwidth]{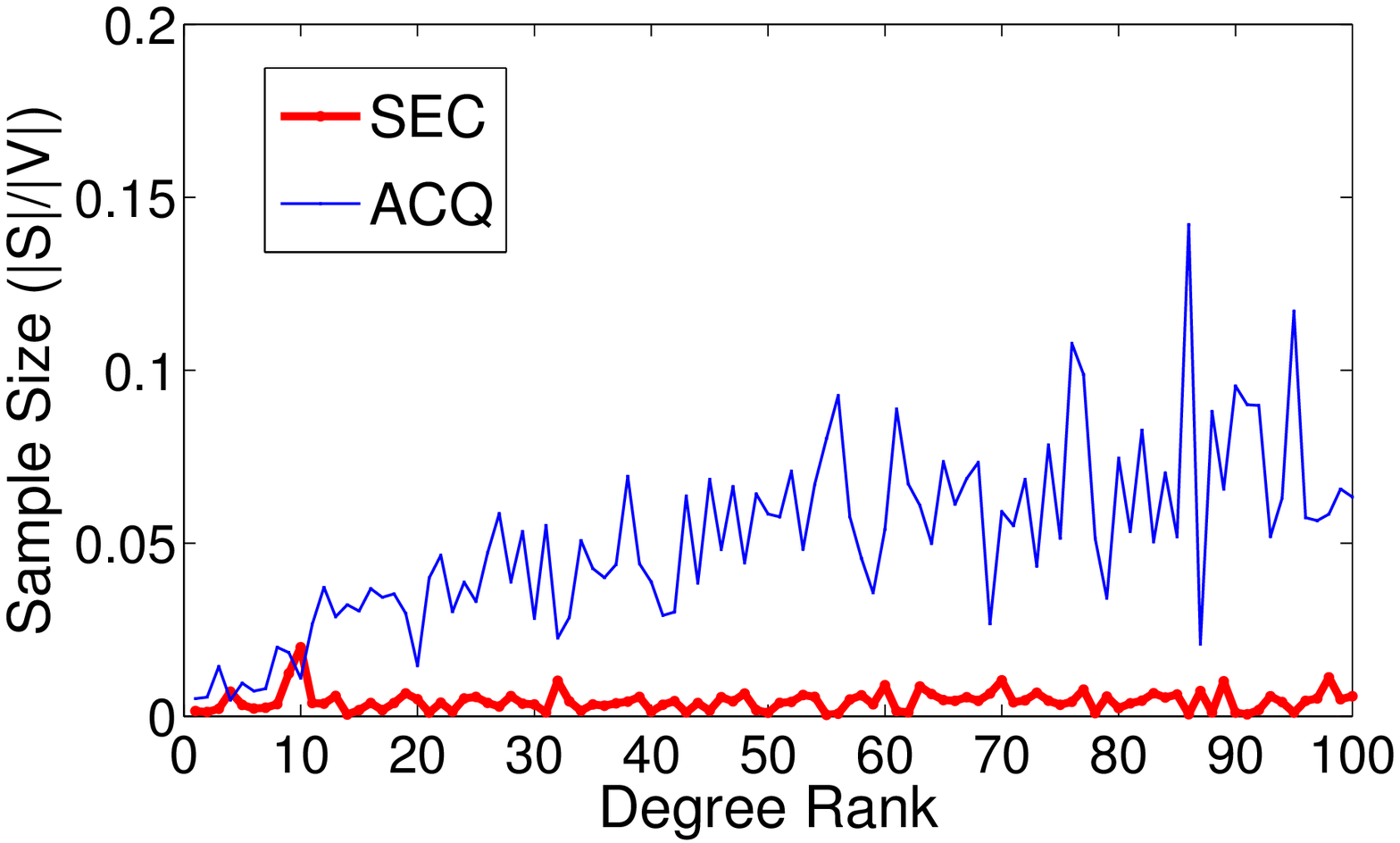}} \vspace{.05cm}
  \subfloat[CondMat]{\label{fig:outdet.condmat}\includegraphics[width=0.2\textwidth]{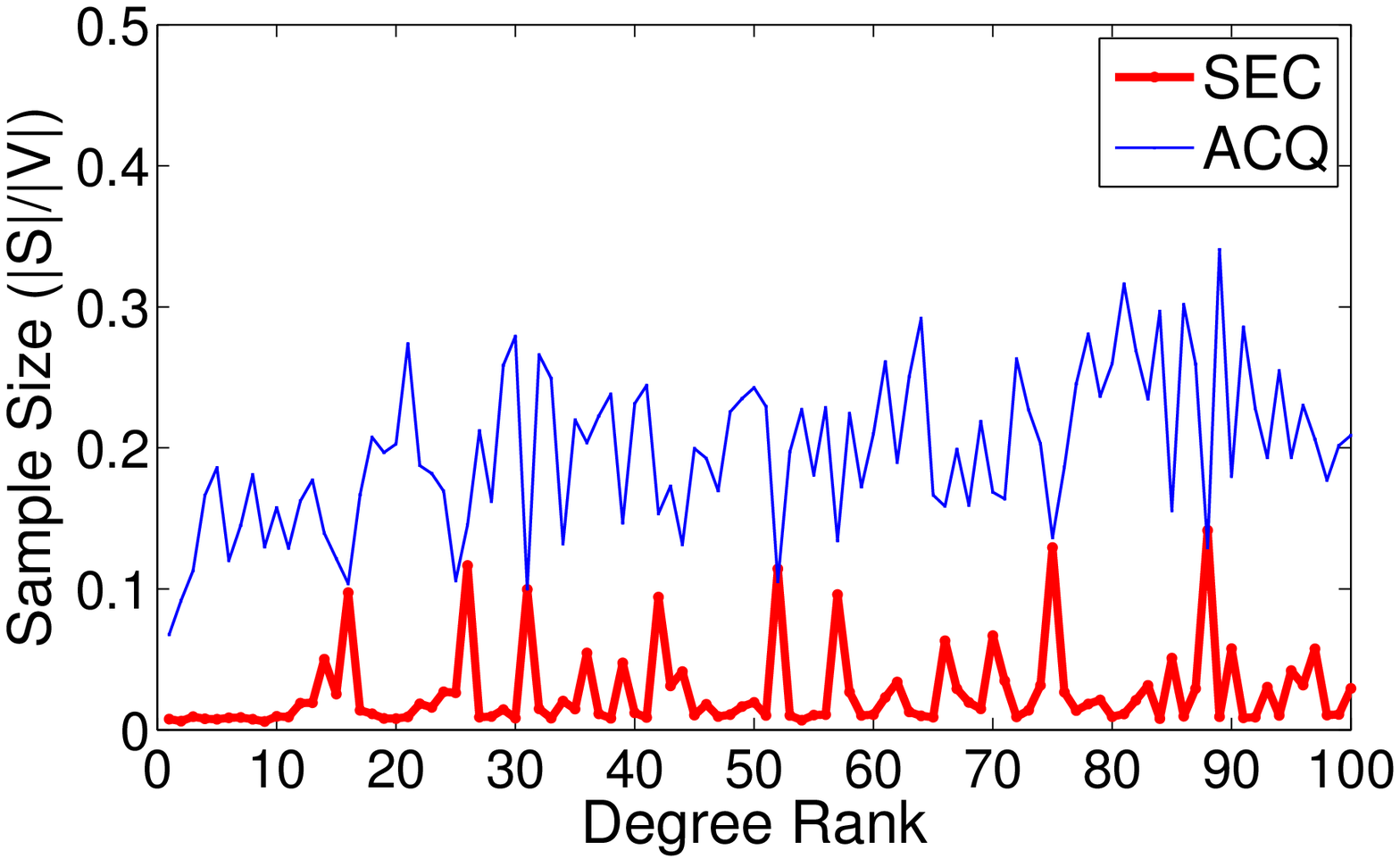}}
\caption{Comparison of SEC and ACQ on quickly locating well-connected individuals (lower is better).  SEC far surpasses ACQ.  Results are similar for every network.}
  \label{fig:outdet}
\vskip -0.15in
\end{figure}

\subsection{Marketing}
Recall from Section \ref{sec:rep.reach} that a community in a network is a cluster of nodes more densely connected among themselves than to others.  Identifying communities is important, as they often correspond to real social groups, functional groups, or similarity (both demographic and not) \cite{Fortunato2010Community}.  The ability to easily construct a sample consisting of members from diverse groups has several important applications in marketing.  Marketing surveys often seek to construct stratified samples that collectively represent the diversity of the population \cite{Kolaczyk2009Statistical}.  If the attributes of nodes are not known in advance, this can be challenging.  The XS strategy, which exhibited the best \emph{community reach}, can potentially be very useful here.  Moreover, it has the added power of being able to locate members from diverse groups with absolutely no \emph{a priori} knowledge of demographics attributes, social variables, or the overall community structure present in the network.  There is also recent evidence to suggest that being able to construct a sample from many different communities can be an asset in effective word-of-mouth marketing \cite{Cao2010OASNET}.  This, then, represents yet another potential marketing application for the XS strategy.

\subsection{Landmarks and Graph Exploration}
\emph{Landmark-based methods} represent a general class of algorithms to compute distance-based metrics in large networks quickly \cite{Potamias2009Fast}.  The basic idea is to select a small sample of nodes (i.e. the landmarks), compute offline the distances from these landmarks to every other node in the network, and use these pre-computed distances at runtime to approximate distances between pairs of nodes.  As noted in \cite{Potamias2009Fast}, for this approach to be effective, landmarks should be selected so that they \emph{cover} significant portions of the network.  Based on our findings for \emph{network reach} in Section \ref{sec:rep.reach}, the XS strategy overwhelmingly yields the best \emph{discovery quotient} and covers the network significantly better than any other strategy. Thus, it represents a promising landmark selection strategy.  Our results for the \emph{discovery quotient} and other measures of \emph{network reach} also yield important insights into how graphs should best be explored, crawled, and searched.  As shown in Figure \ref{fig:rep.reach}, the most prevalently used method for exploring networks, BFS, ranks low on measures of \emph{network reach}.  This suggests that the BFS and its pervasive use in social network data acquisition and exploration (e.g. see \cite{Mislove2007Measurement}) should possibly be examined more closely.

\section{Conclusion}

We have conducted a detailed study on sampling biases in real-world networks.  In our investigation, we found the BFS, a widely-used method for sampling and crawling networks, to be among the worst performers in both discovering the network and accumulating critical, well-connected hubs.  We also found that sampling biases towards high expansion tend to accumulate nodes that are uniquely different from those that are simply well-connected or traversed during a BFS-based strategy. These high-expansion nodes tend to be in newer and different portions of the network not already encountered by the sampling process.  We further demonstrated that sampling nodes with many connections from those already sampled is a reasonably good approximation to sampling high degree nodes.  Finally, we demonstrated several ways in which these findings can be exploited in real-world application such as disease outbreak detection and marketing.  For future work, we intend to investigate ways in which the top-performing sampling strategies can be enhanced for even wider applicability.  One such direction is to investigate the effects of alternating or combining different biases into a single sampling strategy.

\balance
\scriptsize
\bibliographystyle{abbrv}

\end{document}